%% file: main_arxiv_revised2.tex
\documentclass{article}

\usepackage{arxiv}
\usepackage[utf8]{inputenc}
\usepackage{amsmath,amssymb,amsthm}
\usepackage{mathabx}   
\usepackage{graphicx}
\usepackage{booktabs}
\usepackage[hidelinks,colorlinks=false,linkcolor=blue,citecolor=blue]{hyperref}      
\usepackage{xcolor}
\usepackage{natbib}

\newtheorem{theorem}{Theorem}
\newtheorem{corollary}{Corollary}[theorem]
\newtheorem{lemma}[theorem]{Lemma}
\newtheorem*{remark}{Remark}

\let\proglang=\textsf
\newcommand{\pkg}[1]{{\fontseries{b}\selectfont #1}}
\newcommand{\review}[1]{{\leavevmode\color{black}#1}}

\begin{document}

\title{Improving Forecasts for Heterogeneous Time Series by "Averaging", with Application to Food Demand Forecast}

\author{
  Lukas Neubauer \\
  TU Wien\\
  \texttt{lukas.neubauer@tuwien.ac.at} \\
   \And
 Peter Filzmoser \\
  TU Wien\\
  \texttt{peter.filzmoser@tuwien.ac.at} \\
}

\maketitle

% title

\begin{abstract}%
    A common forecasting setting in real world applications considers a set of possibly heterogeneous time series of the same domain. Due to different properties of each time series such as length, obtaining forecasts for each individual time series in a straight-forward way is challenging. This paper proposes a general framework utilizing a similarity measure in \textit{Dynamic Time Warping} to find similar time series to build neighborhoods in a \textit{k-Nearest Neighbor} fashion, and improve forecasts of possibly simple models by averaging. Several ways of performing the averaging are suggested, and theoretical arguments underline the usefulness of averaging for forecasting. Additionally, diagnostics tools are proposed allowing a deep understanding of the procedure.
\end{abstract}

\keywords{Time Series, Forecasting, Combining Forecasts, Dynamic Time Warping, $k$-Nearest Neighbors}

\clearpage

\section{Introduction}

In many forecasting settings, one encounters a heterogeneous set of time series for which individual forecasts are required. A heterogeneous set of time series may imply time series of different seasonalities or trends. Hence, it may be difficult to model all of them in a joint way, or choose an approach yielding reasonable results for the entire set of time series. Modelling each time series on its own might be an alternative way to handle this challenge, however, this does not make use of the shared properties of the domain. Depending on this domain each time series might be very difficult to model and forecast. Additionally, a practitioner may even want an automatic procedure to produce forecasts.

To this end, some work has been done in comparing local and global forecasting approaches, whereas local means modelling each time series by its own. Global models refer to modelling the set of time series simultaneously. \citet{local_vs_global} argue that global and local models can achieve the same forecasts without needing additional assumptions. The global approach does, however, require a rather high number of total observations, additional tuning of hyperparameters, and might be very difficult to find in the first place. Especially the number of total observations can, in practice, be not sufficiently large. On the other hand, local models are especially hard to use for short time series, thus we want to fill the gap of forecasting a set of heterogeneous time series containing also short time series with a still rather low number of total observations.

\citet{local_vs_global_dgp} investigate further the notion of global models whereas they closely look at the relatedness of the time series. The authors focus on simulation experiments controlling the data-generating process of each time series to allow for arbitary relatedness in the set of time series as well. They then apply various machine learning methods to conclude that the performance and complexity of a global models heavily depends on the heterogeneity of the data as well as the actual amount of available data. A set of very heterogeneous time series requires possibly very complex global models which in turn also requires a lot of overall data. This leads us again to the challenge mentioned above.

The shortcomings of a totally global model are addressed by \citet{GODAHEWA2021107518}. The authors argue that with increasing amount of data, a global model may be not localised enough anymore, leading to worse forecasts. Thus, they propose a localisation technique where the time series are clustered and each cluster is then individually modelled by a global model. The clusters are feature- or distance-based, or even randomly assigned. While sounding similar to this paper's work, the approach differs to ours. We do not model each neighborhood of time series - we rather just use the neighbors' existing models. Similarly, \citet{BANDARA2020112896} discuss the use of neural networks on feature-based time series clusters but do require a lot of data. Another approach combining local and global modelling is taken by \citet{SMYL202075}. The author combines neural networks and statistical models such that the neural network is modelled across all time series, and local behavior is modelled by exponential smoothing models. This approach won the M4 forecast competition \citep{M4Comp}. 

In terms of short time series, there is little literature available to tackle this challenge. \citet{short_ts_techniques} compare simple and machine learning based models on a set of short time series ($14$ to $21$ observations each) regarding crimes in Mexico. The authors conclude that simple models like simple moving average or ARIMA perform better than more complex models such as neural networks.

Considering all mentioned aspects, our contributions are as follows. This paper proposes a meta framework for forecasting a set of possibly very heterogeneous time series by utilizing a range of local models and aggregating them in an appropriate way, exploiting similarities of those time series. The similarities can be seen as homogeneous parts of the time series. Hence we do require that there exists some homegeneity in the heteregenous set of time series.

In fact, every time series in a dataset is modelled by a local model. For each time series, we obtain a set of forecasts using the estimated models of its neighboring time series, and perform a type of model averaging to improve its forecasts. This differs from regular model averaging where one time series is modelled by a variety of models. Benefits of this approach are that we firstly allow for simple models in case of short time series, and still do not require a large total number of observations as is the case for complex global models. Therefore, we can say that our methodology is a mixture approach: using local models while still taking into account the whole set of time series in the forecast procedure. This procedure remains very flexible since we do not fix any family of possible models. Even the use of more complex models such as neural networks is still possible.

The rest of the paper is structered as follows. In Section \ref{sec:dtw} we introduce the measure of (dis-)similarity as it is used in our methodology, namely \textit{Dynamic Time Warping} (DTW). More details are available in \ref{sec: dtw2}. Section \ref{sec:knn} introduces the notion of $k$-nearest neighbors in the context of time series while in Section \ref{sec:mavg} we propose several averaging methods to improve forecasts. In \ref{sec: dtw_avg} we shortly outline how DTW can be used to obtain an average time series while a simple theoretical motivation is given in \ref{sec:th_mot} where we validate the notion of model averaging under certain assumptions. Section \ref{sec:eval} is about the evaluation techniques we apply followed by experiments in Section~\ref{sec:experiments} where $5$ different datasets are modelled and deepening diagnostics are performed for one of the datasets. Concluding remarks are given in Section \ref{sec:conc}.

\section{Methodology}
In this paper we propose following methodology. Given a finite set of time series $\mathcal T$ of possibly different properties, we first propose an appropriate dissimilarity measure on $\mathcal T$ by using \textit{asymmetric open-begin open-end Dynamic Time Warping}, see Section \ref{sec:dtw}. This dissimilarity measure is used to construct neighborhoods for each time series in $\mathcal T$. 

Next, for a fixed time series its neighbors are aggregated in a new way to obtain improved forecasts. This aggregation is proposed to be done in multiple ways based on $k$-nearest neighbors (Section \ref{sec:knn}). The actual model-averaging is described in Section \ref{sec:mavg}. Since it might not be clear how DTW corresponds to actual forecasts, we give a theoretical motivation in \ref{sec:th_mot} where the relationship between DTW and forecast distributions is studied for simple state-space models.

A general overview of the methodology given in the form of a flowchart is seen in Figure~\ref{fig:flowchart}.

\begin{figure}[!ht]
    \center
    \includegraphics[width=1.0\textwidth]{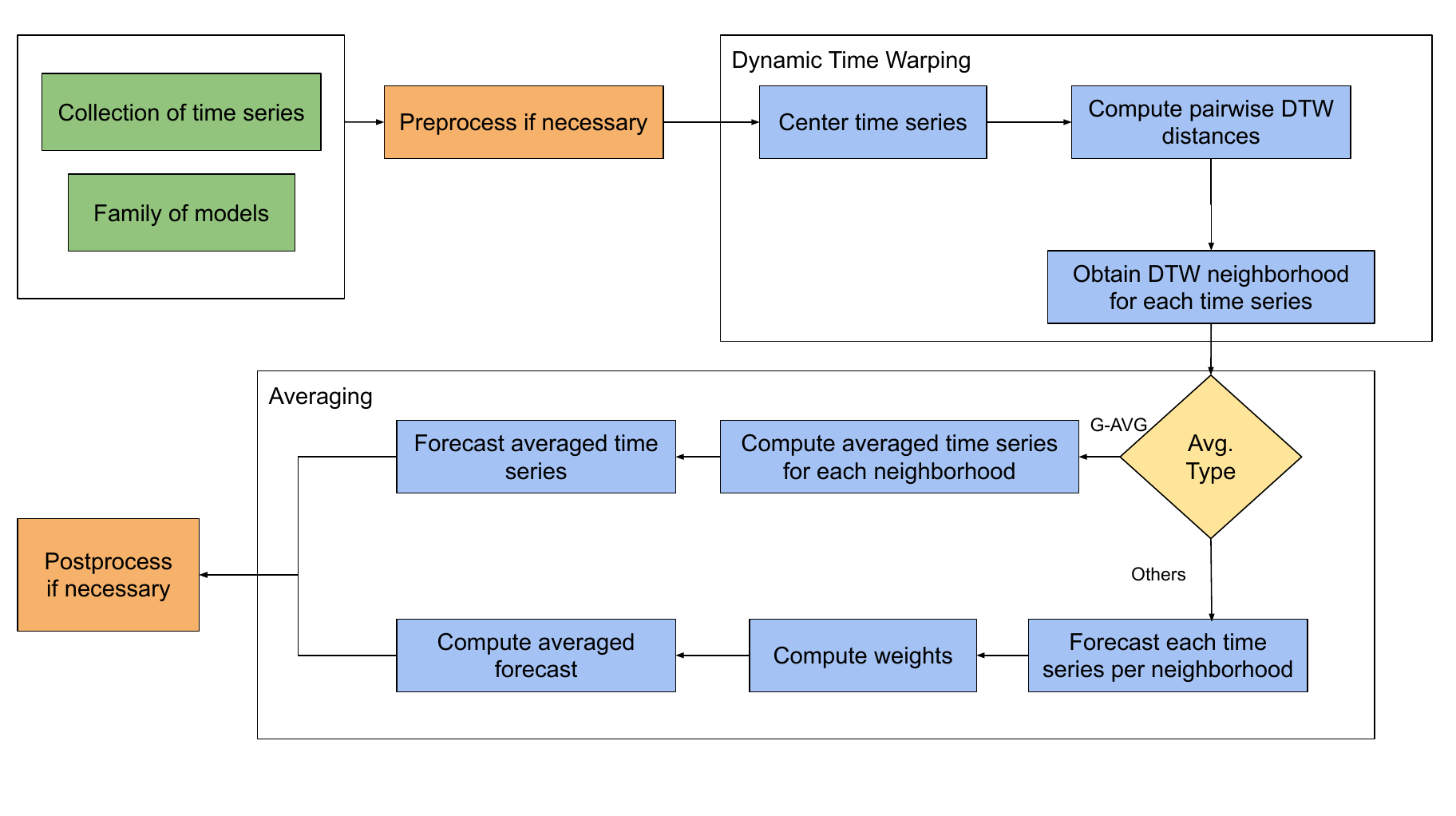}
    \caption{Overall workflow of the proposed averaging methodology.}
    \label{fig:flowchart}
\end{figure}

\subsection{Dynamic Time Warping}\label{sec:dtw}
A common technique for measuring similarity of two time series is \textit{Dynamic Time Warping} (DTW), introduced by \citet{dtw}. Originating from speech recognition, the goal is to align two sequences using a time-warping function in a cost-minimizing manner. In detail, we consider the (possibly multivariate) sequences $\mathbf X=(\mathbf X_1,\dots,\mathbf X_n)'\in\mathbb R^{n\times d}$, $\mathbf Y=(\mathbf Y_1,\dots,\mathbf Y_m)'\in\mathbb R^{m\times d}$ with possibly $n\neq m$ with appropriate dissimilarity function on the rows of $\mathbf X$ and $\mathbf Y$ as $d(i,j):=d(\mathbf X_i,\mathbf Y_j)$. 
%Denote $|\mathbf X|=n,|\mathbf Y|=m$. 
Denote the entries of $\mathbf X_i$ and $\mathbf Y_j$ by
$X_{ik}$ and $Y_{jk}$, $k=1,\ldots ,d$, respectively.
A common choice, also used in this paper, is the Euclidean distance, i.e. $d(i,j)=\sqrt{\sum_{k=1}^d(X_{ik}-Y_{jk})^2}$. The DTW distance is now defined as
\begin{align}\label{eq:dtw}
    \text{DTW}(\mathbf X,\mathbf Y):=\min_{\phi\in\Phi} \sum_{k=1}^{K_\phi}w_\phi(k)d(\phi_{\mathbf X}(k),\phi_{\mathbf Y}(k)),
\end{align}
where $\phi=(\phi_\mathbf X,\phi_\mathbf Y):\{1,\ldots ,K_\phi\}\rightarrow\{1,\ldots ,n\}\times\{1,\ldots ,m\}$ denotes the warping function, and $w_\phi$ are corresponding weights. The warping function $\phi$ links the two sequences in a cost-minimizing way, e.g. $\phi(k)=(\phi_\mathbf X(k),\phi_\mathbf Y(k))$ implies that $\mathbf X_{\phi_\mathbf X(k)}$ and $\mathbf Y_{\phi_\mathbf Y(k)}$ are linked together. The length of the warping function $K_\phi$ depends on the optimal warping function and is therefore chosen alongside the minimization problem in \eqref{eq:dtw}.

A more in-depth introduction of DTW is given in \ref{sec: dtw2}, containing details about the warping functions as well as computing a representative time series using DTW.

\subsection{Nearest Neighbors}\label{sec:knn}

For many classification and regression problems, $k$-nearest neighbors ($k$-NN) is an easy yet well-performing method to apply and improve predictions. The basic idea is as follows. Consider a metric space $(X,\delta)$ and $x\in X$. Then a neighborhood around $x$  can be formed computing $\delta(x,z)$ for each $z\in X,z\neq x$. The $k$ elements with minimal distance to $z$ are then the neighbors of $x$. Eq~\eqref{eq:nhood} gives a formal definition for this.
\begin{align}\label{eq:nhood}
    \mathcal N(x):=\{z_1,&\dots,z_k: \\
    &\delta(x,z_1)\leq\dots\leq \delta(x,z_k)\leq \delta(x,z), z\neq z_i, i=1,\dots,k\}. \nonumber
\end{align}

In case of ties one could increase the neighborhood's size or choose one of the equally distanced neighbors randomly. For the neighborhood including $x$ itself we write
$\overline{\mathcal N}(x):=\mathcal N(x)\cup \{x\}$.

From a statistical point of view, we consider observations $(x_1,y_1),\dots,$ $(x_N,y_N)\in X\times Y$ where $Y$ denotes the set of possible labels and $X$ the feature space. We write $y(x)$ if $y$ is the label corresponding to $x$. In an unsupervised framework there would not be any labels and we would just be able to form neighborhoods. In a classification setting, $Y$ is a discrete set whereas the easiest example is a binary classifier with $Y=\{0,1\}$. A new observation $x_0$ may be classified as the majority class of its neighborhood. In a regression setting, the labels are continuous, e.g. $Y=\mathbb R$ and the predicted label is set to be $\hat y(x_0) = f(\bar{\mathcal N}(x_0))$  where $f$ is an aggregation function. In the simplest case we have $\hat y(x_0)=\frac{1}{k}\sum_{x\in\mathcal N(x_0)}y(x)$, where the predicted label is just the arithmetic mean of the neighborhood's labels.

The choice of $k$ is vital, however, it is still not completely clear how to choose it. There are many heuristics, such as choosing $k\approx\sqrt{N}$, which do not seem reasonable in many applications. In contrast, the optimal $k$ can be chosen based on an optimization criterion. In supervised settings, one usually splits the data in training and test set, and chooses $k$ to minimize a loss function on the training set. This $k$ is then fixed and used to predict on the test set. For more advanced approaches where $k$ is determined adaptively, see \citet{kstar}.

\subsubsection{Time Series Nearest Neighbors}

In the context of time series, $k$-NN has also been widely used. In time series classification, $1$-NN is often used in conjunction with DTW \citep{dtw_classification}. Many approaches also consider $k$-NN in a feature space. Such feature space could consist of time series features like length, trend, auto-correlation properties, and many more. For such cases, DTW is not even used. However, this approach assumes that the features can be extracted in a reasonable way which might not always be the case. In terms of $k$-NN for regression or forecasting, \citet{tsknn} use a simple approach where a single time series can be forecasted by the mean value of neighboring labels based on Euclidean distances of lagged values. Naturally, this method can also be applied to many time series in a pooled manner but this would not use the idea of employing similar time series for improving forecasts, where similarity is based on DTW.

\subsubsection{Our Setting}

We let $X$ be a set of heterogeneous and differently sized time series, equipped with the asymmetric open-begin open-end DTW distance measure. Note that this space is not metric, however, $k$-NN can still be applied in such setting. The label set $Y$ is not uniquely defined. For the \textit{first model, then average} models (see Section \ref{sec:mavg}), this set consists of one-step ahead forecasts for the time series of interest. In the case of \textit{first average, then model}, the label set and the aggregation function are more complicated and are described in the corresponding section.

\subsection{Model Averaging}\label{sec:mavg}

For the upcoming sections, let $\delta(x,y)=\text{DTW}_{\text{asym, OBE}}(x,y)$. Further, let $\hat{\cdot}:(x,M)\mapsto \hat x_{t+1|t,\dots,s}(M)$ with $x=(x_s,\dots,x_t)$ a time series and $M$ any model used to forecast $x_{t+1}$, i.e. $\hat x_{t+1}(M):=\hat x_{t+1|t,\dots,s}(M)$ is the one-step ahead forecast of $x$ obtained by using model $M$. We refer to $M$ as the baseline model which we aim to improve.

We understand a model $M$ as a mapping $M:\mathcal T\rightarrow\mathcal M$, taking a time series $x\in\mathcal T$ and outputting $M_x\in\mathcal M$ where $\mathcal M$ denotes the set of all possible models.

Given a model $M=M_x$ and a time series $z$ we denote $\tilde{\cdot}:(z,M_x)\mapsto \tilde M_x(z)$ as the refitting function, i.e. $\tilde M_x(z)$ is the model $M_x$, estimated on $x$ but fitted on $z$. More specifically, $M_x$ is first trained on $x$ yielding estimated parameters. Fixing those estimated parameters, we now plug in $z$ to obtain fits of the model on $z$. No data of $z$ is used to estimate parameters. Thus, forecasts of $\tilde M_x(z)$ are always with respect to $z$.

\subsubsection{First Average, Then Model Approach}\label{sec: fatm}

The first approach of improving a single forecast yielded by a baseline model is as follows. Given a time series $y=(y_s,\dots,y_t)$, we find a neighborhood $\mathcal N_y=\{z_1,\dots,z_{k_y}\}$ based on $k$-NN. Since the DTW distance also depends on the difference of level in the time series we center each time series first to avoid this effect. That way we keep other properties of the time series such as trends or seasonalities since these might still be helpful to improve forecasts.

Denote $z^c:=(z_{s_z}-\bar z,\dots,z_{t_z}-\bar z)$ the centered time series for $z = (z_{s_z},\dots,z_{t_z})$ with $\bar z=\sum_{i=s_z}^{t_z}z_i/(t_z-s_z+1)$. Then a neighborhood around $y$ is constructed such that $\delta(y^c ,z_1^c)\leq \delta(y_c,z_2^c)\leq\dots\leq \delta(y,z_{k_y}^c)$ and $z_i = (z_{s_i},\dots,z_{t_i})$ with $t\geq t_i,~t_i-s_i>t-s$. That means we look for neighboring time series which have more historical data. This is especially motivated by possibly having shorter time series which will be difficult to model and is often the case in many on-line forecasting settings where data of a subset of the time series are rare. In such case it is not reasonable to allow for even shorter time series as neighbors. Naturally, we also allow for neighbors that might not have recent data. A matching with such time series can still be reasonable due to seasonality effects. Also, not taking such time series into consideration might remove important information from the dataset.

Next, as mentioned in \ref{sec: dtw_avg}, we compute the neighborhood's averaging time series on centered sequences $z^c,z\in\overline{\mathcal N}_y$, given by $\text{avg}(y):=\text{aDBA}(\{z^c:z\in\overline{\mathcal N}_y\})$. 
As the name already suggests, we now take the averaging time series $\text{avg}(y)$ and map it to $M_{\text{avg}(y)}$.

After estimating model $M_{\text{avg}(y)}$ and its parameters we take this model and use it to forecast on $y$, i.e. we obtain $\hat y_{t+1}(\tilde M_{\text{avg}(y)}(y))$. Note that since $M_{\text{avg}(y)}$ is based on centered data, we allow the initial parameters of $M_{\text{avg}(y)}$ to be re-estimated using data of $y$. This procedure automatically handles the re-centering. This averaging will be later referred to as \textbf{G-AVG}.

In terms of $k$-NN we can write the aggregation function as
\begin{align}\label{eq:fatm}
    \hat y_{t+1}&=f(\bar{\mathcal N}(y))\nonumber\\
    &=\hat{\cdot}(y,\cdot)\circ\tilde{\cdot}(y,\cdot)\circ M\circ\text{aDBA}\circ\text{center}(\bar{\mathcal N}(y)),
\end{align}
where, with abuse of notation, we first center the time series of the neighboorhood, then average them, model the resulting averaged time series, and finish by refitting and performing the actual forecast with respect to $y$.

\subsubsection{First Model, Then Average Approach}\label{sec: fmta}

As noted previously, for a time series $y=(y_s,\dots,y_t)$ we obtain a neighborhood $\mathcal N_y$ with size $k_y$. 
However, instead of averaging the neighboring time series in terms of DTW, we consider a different approach. As mentioned, each time series $z\in\overline{\mathcal N}_y$ has been modelled, i.e. there exists $M_z$ for any $z\in\overline{\mathcal N}_y$.

In terms of $k$-NN, the averaging now is of the form
\begin{align}\label{eq:fmta}
    \hat y_{t+1}&=f(\bar{\mathcal N}(y))\nonumber\\
    &=g_w\circ\hat{\cdot}(y,\cdot)\circ\tilde{\cdot}(y,\cdot)\circ M(\bar{\mathcal N}(y)),
\end{align}
where $g_w$ is a weighted averaging function, defined as $g_w(x_1,\dots,x_n):=\sum_{i=1}^nw_ix_i$ where $n$ denotes the number of forecasts to be averaged. We have either $n=k_y+1$ or $n=k_y$.
Compared to the G-AVG approach, we clearly see based on Eq.~\eqref{eq:fatm} and Eq.~\eqref{eq:fmta} that the approaches basically differ in when averaging is performed.

Next, we propose various ways to choose the weights since it is not clear how to choose them in an optimal way.

\subsubsection{Simple Average}

The easiest and straightforward way to combine the forecasts is by taking the simple average. We obtain $w_z=w_y=\frac{1}{k_y+1}$. This assumes that all forecasts are equally important. We will later refer to this type of averaging as \textbf{S-AVG}. Alternatively, we can opt not to use $M_y$ but only utilize the models $\tilde M_z(y)$ of the neighbors, and set $w_z=1/k_y,w_y=0$ (\textbf{S-AVG-N}).

\subsubsection{Distance-weighted Average}

For the distance-based model average we consider two ways. First, we can take the DTW distances of the neighborhood and set 
\begin{align}\label{fc:weights:dist}
    w_z=\cfrac{\cfrac{1}{\delta(y^c,z^c)}}{\sum_{\tilde z\in\mathcal N_y}\cfrac{1}{\delta(y^c,\tilde z^c)}},\quad z\neq y.
\end{align}
implying that closer neighbors in terms of DTW are assigned higher weights in the model averaging. Since $\delta(y,y)=0$, we have to set $w_y=0$ in Eq.~\eqref{fc:weights:dist}, implying that this type of averaging only regards neighbors' models. We denote this type of averaging as \textbf{D-AVG-N}. This averaging does not take into account the forecasts $\hat y_{t+1}(M_y)$ because $w_y=0$. 

In contrast, we can consider the neighborhood's average $\text{avg}(y)$ and corresponding distances $\delta(z^c,\text{avg}(y))$ for $z\in\overline{\mathcal N}_y$. Due to the averaging algorithm (see \ref{sec: dtw_avg}) and a minimum neighborhood size of $1$, we have that $\delta(z^c, \text{avg}(y))>0$ for all $z$. That way we can set up weights as in Eq.~\eqref{fc:weights:dist2}.
\begin{align}\label{fc:weights:dist2}
    w_z&=\cfrac{\cfrac{1}{\delta(z^c,\text{avg}(y))}}{\sum_{z\in\overline{\mathcal N}_y}\cfrac{1}{\delta(z^c,\text{avg}(y))}},\quad z\in\overline{\mathcal N}_y.
\end{align}
We denote this type of averaging by \textbf{D-AVG}.

\subsubsection{Error-weighted Average}

The error-based or, equivalently, performance-based weights are computed from the models' historical performance. They might be calculated in two ways. First, for each time series $z\in\overline{\mathcal N}_y$ we obtain residuals $r_z$ based on the one-step ahead forecasts $\hat z(M_z)$, i.e. $r_{z,i}:=z_i-\hat z_i(M_z)$ for $i=s_z+1,\dots,t_z$. 

Next, we calculate an error measure on the one-step ahead errors, denoted by $E(z,\hat z(M_z))$. Corresponding weights for the model averaging are set to be
\begin{align}\label{fc:weights:pavg}
    w_z=\cfrac{\cfrac{1}{E(z,\hat z(M_z))}}{\sum_{\tilde z\in\bar{\mathcal N}_y}\cfrac{1}{E(\tilde z, \hat{\tilde z}(M_{\tilde z})) }},\quad z\in\overline{\mathcal N}_y.
\end{align}
This method will be referenced to as \textbf{P-AVG}.
In constrast to Eq.~\eqref{fc:weights:pavg}, we consider the refitted residuals $r_y(\tilde M_z(y))$ with $r(\tilde M_z(y))_i := y_i - \hat y_i(\tilde M_z(y))$ for $i=s+1,\dots,t$ and $z\in\overline{\mathcal N}_y$. Consequently, we obtain errors of $E(y, \hat y(\tilde M_z(y)))$, 
and corresponding weights as given in Eq.\eqref{fc:weights:pavgr}.
\begin{align}\label{fc:weights:pavgr}
    w_z=\cfrac{\cfrac{1}{E(y, \hat y(\tilde M_z(y)))}}{\sum_{\tilde z\in\bar{\mathcal N}_y}\cfrac{1}{E(y, \hat y(\tilde M_{\tilde z}(y)))}}\quad z\in\overline{\mathcal N}_y.
\end{align}
We denote this method by \textbf{P-AVG-R}. In practice, we utilize the root mean squared scaled error measure (RMSSE) as introduced in Section~\ref{sec:eval}.

\subsubsection{No-Model Approaches}

In respect to the non-parametric neighboring search of k-NN, we also consider non-parametric forecasts as follows. Following \citet{tsknn} where forecasts are computed based on the next available observed value, we apply a similar methodology. For the given time series $y=(y_s,\dots,y_t)$ and neighbors $z=(z_{s_z},\dots,z_{t_z})\in\mathcal N_y$, we obtain matchings (on which the neighborhood is based on) such that $\phi^{(z)}(k)=(\phi_y(k),\phi_z(k))$ for $k=1,\dots,K_z$. Since every index of $y$ must be matched, there exists a value $\tilde k$ and $\tilde t_z\in\{s_z,\dots,t_z\}$ such that $\phi^{(z)}(\tilde k)=\left(t,\tilde t_z\right)$. Next, there are two cases to differ.
\begin{itemize}
    \item Case $\tilde t_z < t_{z}$: Then there exists a successor $u_z=\tilde t_z+1$ with corresponding value of $z_{u_z}$ which is used to forecast $y_{t+1}$.
    \item Case $\tilde t_z = t_{z}$: No successor exists, thus the time series $z$ cannot be used to forecast $y_{t+1}$.
\end{itemize}
To this end, we obtain the set of possible successors given by
\begin{align}\label{fc:nm_succ}
    \mathcal S_y = \left\{z_{u_z}: u_z = \tilde t_z+1, \tilde t_z <t_z, z\in\mathcal N_y\right\}.
\end{align}

Clearly, Eq.~\eqref{fc:nm_succ} implies that $|\mathcal S_y|\leq |\mathcal N_y|$. This type of averaging also only considers neighbors' information because $y$ does not have any successor itself.

Similarly to before, we can forecast $y_{t+1}$ by $\hat y_{t+1}=\sum_{s\in\mathcal S_y}w_s s$ with approriate weights. We choose simple weights as in $w_s=1/|\mathcal S_y|$ (\textbf{S-NM-AVG}) and distance-based weights as in Eq. \eqref{fc:weights:dist} (\textbf{D-NM-AVG}). Note that the matchings are based on centered time series, hence the forecasts obtained here are actually forecasts for $y^c$. Since we do not have a model taking care of re-centering, we have to do it manually by setting $\hat y_{t+1}=\hat y_{t+1}^c+\bar y$.

    The methodology's notation is briefly summarised in Table~\ref{tab:abbr}.
    \input{paper_tables/abbr.tex}

\iffalse
\remove{
\textbf{extension to obe, to also allow differently sized processes}
When considering sequences of different lengths, we are limited because once $Y$ is twice as long as $X$ or even longer, the recursion cannot be calculated anymore as Eq. \eqref{eq:dtw_ann_proof} suggests. This happens quite often in practice, hence we use the \textit{open-begin open-end} matching procedure. Similarly to before, we obtain
\begin{corollary}
    Let $X\sim ANN(\alpha_X,\sigma^2_X),Y\sim ANN(\alpha_Y,\sigma^2_Y)$ be two independent and centered Exponential Smoothing processes of lengths $n$ and $m$. Let $d(i,j):=\mathbb E[(X_i-Y_j)^2]$ denote the cross-distance between $X$ and $Y$.\\
    Then the asymmetric DTW distance for any $1\leq p\leq q\leq m$ is given by
    \begin{align}\label{eq:dtw_obe_ann}
        \text{DTW}_\text{asym}(X,Y^{(p,q)})=\sigma^2_X\left(n+ {\binom{n}{2}} \alpha_X^2\right)+\sigma^2_Y \left(n+ c(p,q)\alpha_Y^2\right),
    \end{align}
    with a constant $c(p,q)$ defined by
    \begin{align*}
        c(p,q)=\begin{cases}
            p-1+\cfrac{(q-p)^2}{4} & q-p\text{ even,}\\
            1+\cfrac{1+q-p^2}{2}+\cfrac{q^2-(p+1)^2}{4} & q-p\text{ odd.}
        \end{cases}
    \end{align*}
\end{corollary}
miniming $DTW_{asym}(X,Y^{(p,q)})$ yields $p=q=1$ since the constant in minimized there. Unfortunately this does not make much sense in practice.
}
\fi
\section{Evaluation Methods}\label{sec:eval}

For the evaluation of our methods we use scaled one-step ahead forecast errors as proposed by \citet{scaledrmse}, adapted to our problem setting of on-line forecasting. The authors denote the mean absolute scaled error (MASE) by
\begin{align}\label{eq:mase}
    \text{MASE}(y,\hat y) &= \cfrac{1}{t-s}\sum_{u=s+1}^{t} |q_u|,\\
    q_u &= \cfrac{y_u-\hat y_u}{\frac{1}{t-s}\sum_{v=s+1}^{t}|y_v-y_{v-1}|}, \nonumber
\end{align} 
where $y=(y_s,\dots,y_t)$ is a time series with one-step ahead forecasts $\hat y=(\hat y_{s+1|s},\dots,\hat y_{t|t-1})$. This means we scale each error by the average in-sample error when using the last available observation as the one-step ahead forecast, also known as the random walk forecast.

As mentioned by \citet{scaledrmse}, one advantage of this measure is its independence of scale, making it easier to compare results of different time series. Since the measure compares the actual forecast with the mean forecast error based on a random walk forecast, we can say that if $\text{MASE}<1$, then the forecast method used to obtain $\hat y$ works better on average than the naive approach of using the last available value. Similarly, if $\text{MASE}>1$, then the method performs worse then the random walk forecast.

However, in an on-line forecasting setting, scaling by the whole in-sample error may not be reasonable, hence we use a different scaling. We adapt Eq.~\eqref{eq:mase} to a root mean squared scaled error (RMSSE) given by
 
\begin{align}\label{eq:rmsse2}
    \text{RMSSE}_{s'}^{t'}(y,\hat y,\hat y^b)&:=
    \sqrt{\cfrac{1}{t'-s'+1}\sum_{u=s'}^{t'}q_u^2},\\
    q_u &= \cfrac{y_u-\hat y_u}{\sqrt{\frac{1}{u-s}\sum_{v=s+1}^{u}(y_v-\hat y^b_v)^2}},\nonumber
\end{align}

where $s\leq s'\leq t'\leq t$. This means we obtain scaled errors $q_u$ by scaling the error $y_u-\hat y_u$ by the averaged benchmark forecast error until time $u$. Altogether, $\text{RMSSE}_{s'}^{t'}$ in Eq.~\eqref{eq:rmsse2} gives the average error of the window $\{s',\dots,t'\}$ scaled by the average historical benchmark forecast until time $t'$. A typical benchmark method is the random walk forecast given by $\hat y_v^b = y_{v-1}$.

\iffalse
\remove{
For evaluation purposes, we will also look at a differently scaled RMSSE value, yielding simpler comparisons between the \textbf{ETS} models and averaging approaches. To this end, we define
\begin{align*}
    \tilde q_u = \cfrac{y_u-\hat y_u}{\sqrt{\frac{1}{u-s}\sum_{v=s+1}^{u}(y_v-\hat y_v(M_y))^2}},
\end{align*}
and an aggregated value of
\begin{align}
    \widetilde{\text{RMSSE}}_{s'}^{t'}(y,\hat y)&=
    \sqrt{\cfrac{1}{t'-s'+1}\sum_{u=s'}^{t'}\tilde q_u^2}.
\end{align}
}
\fi
%\subsection{Evaluation on Train Set}

As the data is usually split into training and test set, we also need to differ the corresponding evaluation techniques. In the training set we compute the RMSSE as given above, yielding a final RMSSE value for the entire training set.

%\subsection{Evaluation on Test Set}

For an evaluation on the test set
we follow the notation of \citet{scaledrmse} and its use in the \proglang{R} package \pkg{forecast} \citep{forecast-pkg} where the test errors are scaled by the training error of the random walk. Consider a time series $y=(y_s,\dots,y_t,y_{t+1},\dots,y_T)$ split in training and test set at time step $t$. Then each test set residual is scaled by the root mean squared error of the random walk forecast yielding scaled test set residuals given by Eq.~\eqref{eq:scaled_errors_test}.
\begin{align}\label{eq:scaled_errors_test}
    q_u = \cfrac{r_u}{\sqrt{\frac{1}{t-s}\sum_{v=s+1}^{t}(y_v-\hat y^b_v)^2}},\quad T\geq u > t.
\end{align}

Additionally to RMSSE, we also consider more common forecasting measures such as mean absolute error (MAE) \eqref{eq:mae}, root mean squared error (RMSE) \eqref{eq:rmse} and symmetric mean absolute percentage error (sMAPE) \eqref{eq:smape}, see \cite{scaledrmse} for a recent review. 
\begin{align}
    \text{MAE} &= \cfrac{\sum_{i=t+1}^{T}|y_i-\hat y_i|}{T-t} \label{eq:mae}\\
    \text{RMSE} &= \sqrt{\cfrac{\sum_{i=t+1}^{T}(y_i-\hat y_i)^2}{T-t}} \label{eq:rmse}\\
    \text{sMAPE} &= \cfrac{2}{T-t} \sum_{i=t+1}^{T}\cfrac{|y_i-\hat y_i|}{|y_i|+|\hat y_i|} \label{eq:smape}
\end{align}

For each dataset the mean and median error for each measure is computed.

\section{Experiments}\label{sec:experiments}
\review{
To demonstrate the methodology, we use $5$ publicly available datasets. Table \ref{tab:datasets} provides summary information of these datasets. 
}

\input{paper_tables/datasets}
All datasets are on monthly basis except for the Food Demand data which is a weekly dataset. A short summary is as follows.
\begin{itemize}
    \item \review{Food Demand. Weekly data from smart fridges with the goal to forecast the demand of each fridge for the upcoming week in one-step ahead fashion. The dataset is available in this paper's \proglang{R} package.}
    \item M3 Dataset \citep{M3Data}. 
    Monthly data of the M3 competition, split in $5$ subcategories, namely \textit{micro}, \textit{macro}, \textit{industry}, \textit{demographic} and \textit{finance}. The $6$-th subcategory \textit{other} was removed from the analysis. We use this prior domain knowledge and model each subcategory individually. The data was obtained from the \proglang{R} package \pkg{Mcomp} \citep{M3DataR}.
    \item Tourism \citep{TourismData}. 
    Monthly dataset from the tourism forecast competition. The data is available from the \proglang{R} package \pkg{Tcomp} \citep{TourismDataR}.
    \item CIF 2016 \citep{CIF2016Data}. Monthly data from the CIF 2016 forecasting competition. This data contains $24$ real time series from the banking domain and $48$ artifically created time series. The dataset is available from Zenodo \citep{CIF2016Zenodo} and is put in this paper's corresponding \proglang{R} package for convenience.
    \item Hospital. Monthly time series counting patients for different medical products and related medical problems. The dataset is available in the \proglang{R} package \pkg{expsmooth} \citep{expsmooth-pkg}.
\end{itemize}
\review{All datasets in Table~\ref{tab:datasets} have a forecast horizon greater than $1$, whereby the Food Demand data is a special case since we are only interested in one-step ahead forecasts for that dataset.} Hence we run two experiments on these datasets: multi-step ahead forecasts as required for the datasets, and cumulative one-step ahead forecasts to compare to the Food Demand data results. We do not compute multi-step forecasts for the Food Demand data.

We denote by TSAVG the trained averaging method. For all datasets the same training procedure was performed.

\subsection{Baseline Model}\label{sec:ets}

For the base family of models, any common time series model family can be used. We here use the {ETS} family of models \citep{expsmoothing_hyndman,expsmoothing_taylor}. These models consist of $3$ components: \textbf{E}rror, \textbf{T}rend and \textbf{S}easonality. 

The most simple ETS model is of the form $ANN$ and is also known as \textit{Exponential Smoothing}. It does not have any trend or seasonality component, and is given by the recursion
\begin{align}\label{eq:ann1}
    l_t &= \alpha y_t + (1-\alpha)l_{t-1},\\
    \hat y_{t+h|t} &= l_t,\quad h>0,\nonumber
\end{align}
where $l$ denotes the level component of the model which is also equal to the flat forecast $\hat y_{t+1|t}$. The model parameter $\alpha$ is usually found by minimizing the sum of squared one-step ahead forecast errors.

The next more complex model is of the form $AAN$ and is called \textit{Holt-Winters} model. It is given by
\begin{align}\label{eq:aan}
    l_t &= \alpha y_t + (1-\alpha)(l_{t-1}+b_{t-1}),\\
    b_t &= \beta (l_t-l_{t-1}) + (1-\beta)b_{t-1},\nonumber\\
    \hat y_{t+h|t} &= l_t+hb_t,\quad h>0, \nonumber
\end{align}
where $l$ denotes the level component again. The newly introduced trend component is denoted by $b$. In this model the forecasts are not flat anymore, but linear in the forecast horizon $h$. All smoothing parameters are usually constrained between $0$ and $1$. For a more detailed review on those models and many more equations like Eq.~\eqref{eq:ann1} and Eq.~\eqref{eq:aan}, see \citet{exp_smoothing_review}.

With that, this family provides a very flexible way of modelling and forecasting, and is therefore often used in business forecasting applications. The \pkg{forecast} package \citep{forecast-pkg} in \proglang{R} also provides an automatic forecasting framework for {ETS} and other models (such as ARIMA), where the most appropriate model of all possible models in that family is chosen automatically based on some optimization criteria. We use the \textit{corrected Akaike's Information Criterion} given by
\begin{align}\label{eq:aicc}
    \text{AIC}_c&=\text{AIC}+\cfrac{2k(k+1)}{T-k-1}
    =-2\log(L)+2k+\cfrac{2k(k+1)}{T-k-1},
\end{align}
where $L$ is the model's likelihood, $k$ denotes the total number of parameters in the model, and $T$ is the sample size. The correction is needed because in small samples the regular AIC tends to overfit and select models with too many parameters \citep{aicc}.

Each time series is modelled by ETS and then forecasted in a one-step ahead fashion, irregarding its length. 

Additionally, another common model family of time series models in ARIMA \citep{Shumway2000} is used to model the dataset, again utilizing the \pkg{forecast} package. As for the ETS family, the best model is chosen by the corrected AIC as defined in Eq.~\eqref{eq:aicc}.

\subsection{Global Benchmark Model}

As a second benchmark model we use a pooled auto-regressive (AR) model \citep{BALTAGI2021}, which can be seen as a global model since it uses all individuals' information. This linear $AR(p)$ model is given by 
\begin{align}\label{eq:ar_p}
    y_{i,t}=\alpha+\sum_{k=1}^p\beta_k y_{i,t-k}+\epsilon_{i,t},
\end{align}
where $i=1,\dots,N$ denotes the $i$-th individual and $t=s_i,\dots,t_i$ denotes the time component of the panel. The order of the model is given by $p\geq 1$.
This means that for the entire panel data we only need to estimate $p+1$ parameters in Eq.~\eqref{eq:ar_p} - the global intercept $\alpha$ as well as the global slope parameters $\beta_1,\dots,\beta_p$. In practice, the observations of all individuals are stacked on top of each other to obtain a simple linear model which can be estimated by ordinary least squares.

Such models, also called panel models, can be of different forms as well. We also modelled the data using variable intercepts, i.e. each individual has an individual intercept. However, this model turned out to be worse than the pooled model for the datasets.

Another variant is the variable-coefficient model with individual intercept and slope. However, such model is usually estimated by estimating each individual's model by its own. Thus, such approach cannot really be seen as a global model, and thus we opted to not use this model as well.

For more details about the models and their statistical properties, see the book of \citet{BALTAGI2021}. All global panel models have been fitted using the \pkg{plm} package \citep{plm_package} in \proglang{R}. \review{Lag selection for $p=1,\dots,L_\text{max}$ was performed using time series cross-validation  as in Section~\ref{sec:tscv}. For each dataset an individual maximum lag was selected based on the rule of thumb of \citep{BANDARA2020112896}, namely $\text{input size}=1.25\max(\text{output size},\text{seasonal period})$. The specific values of $L_\text{max}$ are given in Table~\ref{tab:train_params} while $5$ equidistant values for $p$ are chosen for each dataset.}

\subsection{Choice of Optimal Parameter}\label{sec:tscv}

In every experiment we perform time series cross-validation (TSCV) based on a rolling window approach to choose the optimal hyperparameter $k$ in $k$-NN. We cannot apply regular cross-validation since the assumption of independent data is not valid in a time series context.

The initial window contains observations from time step $1$ until time step $T_0$. In each iteration, the window is expanded by one time step, yielding a new fold. In each fold we obtain an RMSSE value for each individual and hyperparameter. In detail, assume there exists a hyperparameter $k\in\Theta$, where $\Theta$ is the search grid. Let a fold be $f=\{1,\dots,T_0,\dots,t_f\}, T_0\leq t_f\leq T_{\text{train}}$ and consider an individual time series $y^{(i)}=(y_{s_i},\dots,y_{t_i})$. Then we compute $\text{RMSSE}_{\max(T_0,s_i)}^{t_f}(y^{(i)}, \hat y^{(i)}(m); k)$ with random walk benchmark forecasts for each parameter $k$, time step $t_f$, individual $i$ and method $m$. Furthermore, the cross-validation yields standard errors as well by aggregating over the folds, i.e.

\begin{align}\label{eq:se}
    &\text{SE}(y^{(i)},\hat y^{(i)};k) = \sqrt{\cfrac{1}{T_\text{train}-\max(T_0,s_i)+1}}\\
    &\sqrt{\cfrac{\sum_{t_f=\max(T_0,s_i)}^{T_\text{train}}\left(\text{RMSSE}_{\max(T_0,s_i)}^{t_f} (y^{(i)},\hat y^{(i)}(m);k)-\mu(y^{(i)},\hat y^{(i)}(m);\theta)\right)^2}{T_\text{train}-\max(T_0,s_i)}},\nonumber
\end{align}
with the mean cross-validation error, also known as \textit{CV score}, of 
\begin{align}\label{eq:cv_score}
    \mu(y^{(i)},\hat y^{(i)}(m);k)=\cfrac{\sum_{t_f=\max(T_0,s_i)}^{T_{\text{train}}}\text{RMSSE}_{\max(T_0,s_i)}^{t_f} (y^{(i)},\hat y^{(i)}(m);k)}{T_\text{train}-\max(T_0,s_i)}.
\end{align}

Note here that each individual $i$ might be available in a different amount of folds since we perform cross-validation based on the time aspect of the panel data. The optimal hyperparameter is now chosen using the one-standard-error rule, i.e. we choose the most parsimonious model lying in the one standard error band of the global minimum of cross-validation errors values. This means that based on Eq.~\eqref{eq:se} and Eq.~\eqref{eq:cv_score} we obtain the optimal parameter $k^\ast$ for $y^{(i)}$ as in the minimization problem \eqref{eq:k_min} where $k_i^{(m)}$ is the parameter realizing the minimal mean cross-validation error, i.e. $k_i^{(m)}=\text{argmin}_k \mu(y^{(i)},\hat y^{(i)}(m);k)$.

\begin{align}\label{eq:k_min}
    k^\ast_i = \min \left\{\theta:|k-k_i^{(m)}|\leq \text{SE}(y^{(i)},\hat y^{(i)};k_i^{(m)})\right\}.
\end{align}

\review{Training parameters such as $T_0$ are all listed in Table~\ref{tab:train_params} for each dataset. The reason why a mean value is given for $T_\text{train}$ is because of the possibly different lengths of the time series.}

\input{paper_tables/train_params.tex}

\subsection{Results}
Tables~\ref{tab:multistep_mae}, \ref{tab:multistep_rmse}, \ref{tab:multistep_rmsse}, \ref{tab:multistep_smape} show the corresponding mean and median errors for the multi-step horizon case for each dataset. 

\review{
First, we note that the pooled AR model seems to perform well for the M3 and Tourism dataset. Those datasets are characterized by long forecast horizon as well as long time series which could be a reason for the good performance. When comparing the averaging methodology to the local benchmarks ETS and ARIMA, there is no clear picture present. Nevertheless, TSAVG does seem to keep up with the benchmark methods and is able to improve performance of the ETS model which it is supposed to do. For testing statistical significance, non-parametric Friedman tests are performed on the mean values and the corresponding p-values are given in the tables' captions. No statistical significant difference for an $\alpha$-level of $5\%$ is found for any error measure.

The corresponding one-step ahead errors are as seen in Tables~\ref{tab:onestep_mae}, \ref{tab:onestep_rmse}, \ref{tab:onestep_rmsse}, \ref{tab:onestep_smape}. We observe a similar picture that the pooled model performs well for the M3 and Tourism dataset. When comparing TSAVG to ETS and ARIMA, the numbers are very close and no clear best method is present. In the case of one-step ahead errors, TSAVG is only able to improve performance of the ETS model in a few cases. In case of the Food Demand dataset containing just a small number of rather short time series we do observe improvements over the local base models as well as all other competing models. As before, Friedman tests are performed on the mean values. The null hypothesis of no differences can not be rejected for any error measure again. 

To summarize the results, the proposed method is competitive on all datasets. Especially on the Food Demand dataset we observe good performance due to its characteristics. However, TSAVG also performs comparably on datasets which do not share the same properties as the Food Demand dataset. For example, the Hospital dataset only contains equally sized time series whereas the Tourism dataset contains very long time series. This makes TSAVG a very flexible method to improve simple local models for a wide range of datasets.
}

When looking at all results with respect to the datasets' sizes, there does not seem to exist a huge impact. We do, however, observe that the bigger the dataset the better TSAVG seems to perform compared to the baseline ETS model. This is reasonable since the number of possible neighbors increases with increasing dataset size.
\input{paper_tables/multistep_eval_new.tex}
%\begin{landscape}
\input{paper_tables/onestep_eval_new.tex}
%\end{landscape}
\iffalse
\begin{figure}[!ht]
    \center
    \includegraphics[width=\textwidth]{paper_figures/datasets_rmsse_percentiles.pdf}
    \caption{Percentiles plots for each dataset.}
    \label{fig:datasets_rmsse}
\end{figure}
\begin{figure}[!ht]
    \center
    \includegraphics[width=\textwidth]{paper_figures/m3_rmsse_percentiles.pdf}
    \caption{Percentiles plots for each category of the M3 dataset.}
    \label{fig:m3_cat_rmsse}
\end{figure}
\fi
\subsection{Runtime}
Figure~\ref{fig:datasets_runtime} shows training runtimes for a selection of averaging methods. We clearly see the effect of increasing neighborhood size across all datasets and averaging methods. The runtime also increases significantly with the datasets' size due to the pairwise search for DTW neighbors. While the calculation times of P-AVG and S-AVG are similar because the corresponding weights are both quickly computed, the distance-based averaging method has increasing computation time because of needing to compute the neighborhood's averaged time series.
\begin{figure}[!ht]
    \center
    \includegraphics[width=\textwidth]{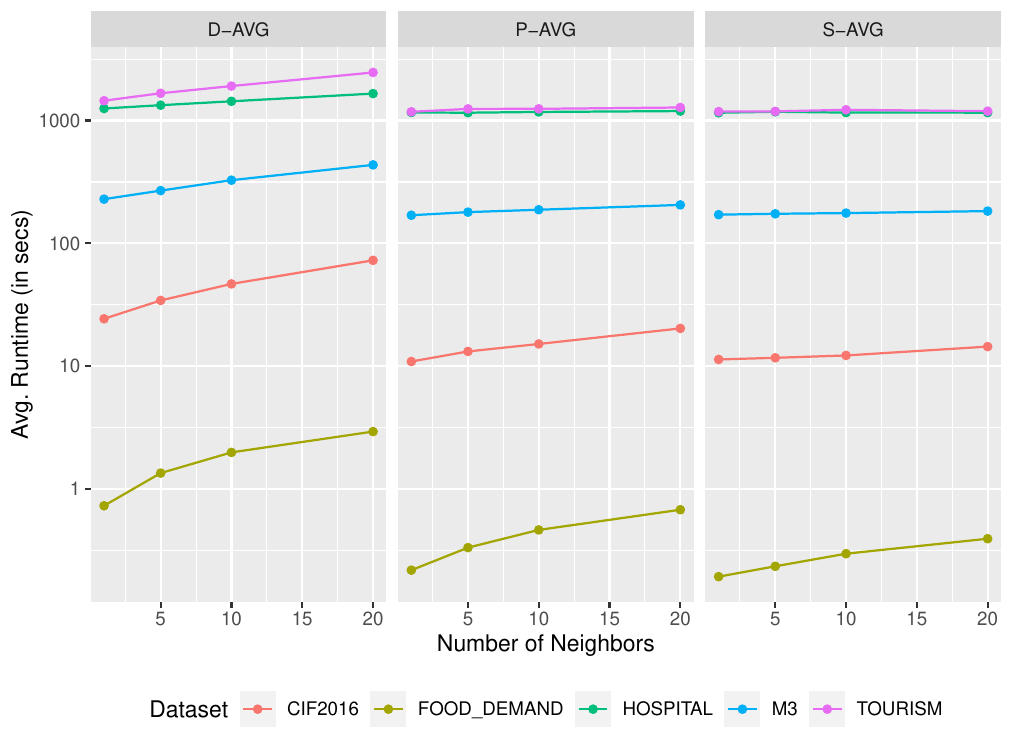}
    \caption{Runtime analysis for each dataset and selected averaging methods. The shown numbers are averaged runtimes of the TSCV procedure.}
    \label{fig:datasets_runtime}
\end{figure}

\review{
\subsection{Diagnostics for the Food Demand Data}\label{sec:data_application}

The transparency of TSAVG allows for advanced diagnostics. These are demonstrated on the Food Demand dataset which is about forecasting weekly demand of smart fridges. Its rather small size as well as short time series lead to understandable diagnostics on an individual time series level. The practical problem is also what initially motivated the idea of using pattern matching to improve forecasts. A selection of example time series of the dataset are available in Figure~\ref{fig:ts_example} and show its heterogeneity. While homogeneity is still present, individual $29$, for example, shows a very similar pattern to the first part of individual $3$. This is the similarity we ought to extract using DTW.
}

\begin{figure}[!ht]
    \center
    \includegraphics[width=0.8\textwidth]{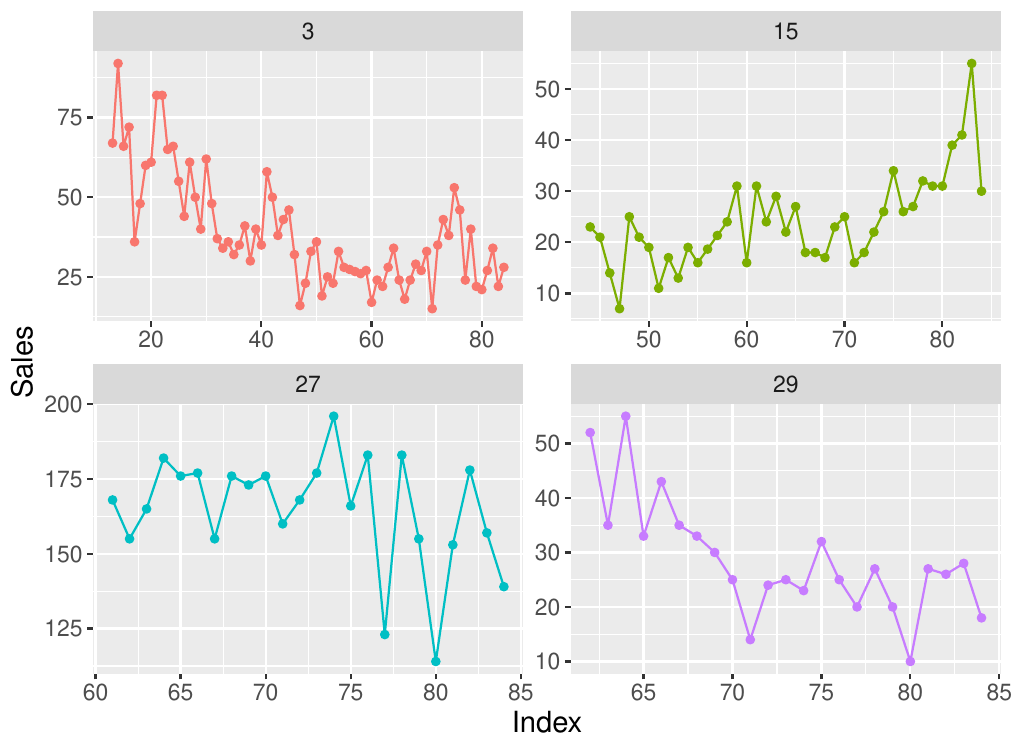}
    \caption{Example time series of the Food Demand dataset.}
    \label{fig:ts_example}
\end{figure}

\subsubsection{TSCV Diagnostics}

In each fold we forecast the time series to obtain one-step ahead forecasts as explained in Section \ref{sec:tscv}. The ETS models themselves are not being tuned using the folds. The best ETS model is solely chosen in each fold using AIC. That way the only tuning parameter is the size of each individual's neighborhood. 

We take a closer look at the TSCV procedure and the choice of the optimal tuning parameters. Figure \ref{fig:cv_dist} shows the corresponding TSCV scores and standard errors as defined above for the distance-based averaging methods. Each plot panel shows the results for a selected smart fridge.
For demonstration purposes we choose a grid of $\Theta=\{1,3,5,10,20\}$ number of neighbors and select the optimal number applying the one standard error rule. We do this for each forecast averaging methodology. The vertical lines indicate the choice of the tuning parameters. We observe different behaviors between the global averaging approach and the others (D-AVG, D-AVG-N). This can be explained by the similarity of the two methods. We also see that the scores for individual $3$ obtain a constant level because this individual is one of the longer time series, and hence it is not even possible to have more than $3$ neighbors.

\begin{figure}[!ht]
    \center
    \includegraphics[width=0.8\textwidth]{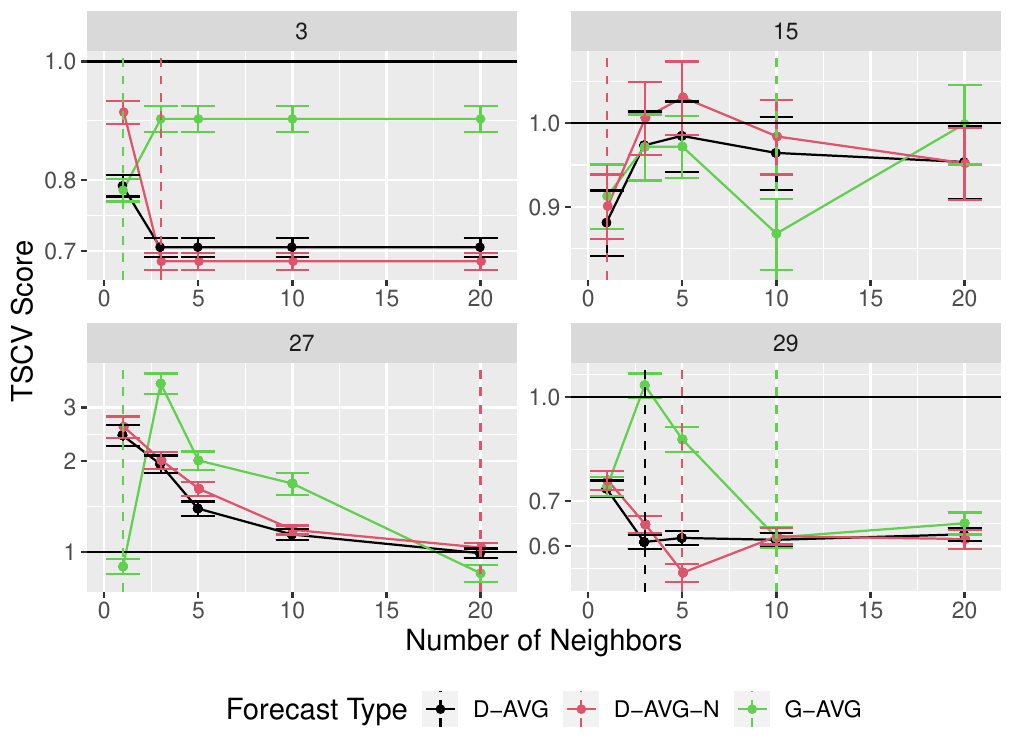}
    \caption{TSCV scores for distance-based averaging. The dashed vertical lines indicate the optimal number of neighbors.}
    \label{fig:cv_dist}
\end{figure}

The actual selected number of neighbors is given in Table~\ref{tab:no_neighbors}. See Table~\ref{tab:abbr} for a description of all averaging approaches.

\input{paper_tables/best_k.tex}

\subsubsection{Model Evaluation}

First, we start off by evaluating the experiment on a global level. This is where the scaled errors come in handy since these are comparable also on individuals' level. Figure~\ref{fig:global_rmsse_f} shows boxplots of the RMSSE values as in Eq.~\eqref{eq:rmsse2}, split in both training and test errors. Given the errors are skewed, we display them on a log-scale. We clearly observe that the simple, non-models based averaging approaches yield worse results than the ETS benchmark model. However, we do see minor improvements of the ETS forecasts by using distance- and performance-based averaging methods. Another interest fact is that approaches not considering an individual's information (S-AVG-N, D-AVG-N) tend to yield also worse results, indicating that these informations should also be included in the forecast procedure, even though these time series might be very short and difficult to forecast. For completeness, Figure \ref{fig:global_rmsse} in the Appendix shows the original RMSSE values, but as mentioned before, this makes comparisons to the ETS benchmark model quite difficult.

\begin{figure}[!ht]
    \center
    \includegraphics[width=0.8\textwidth]{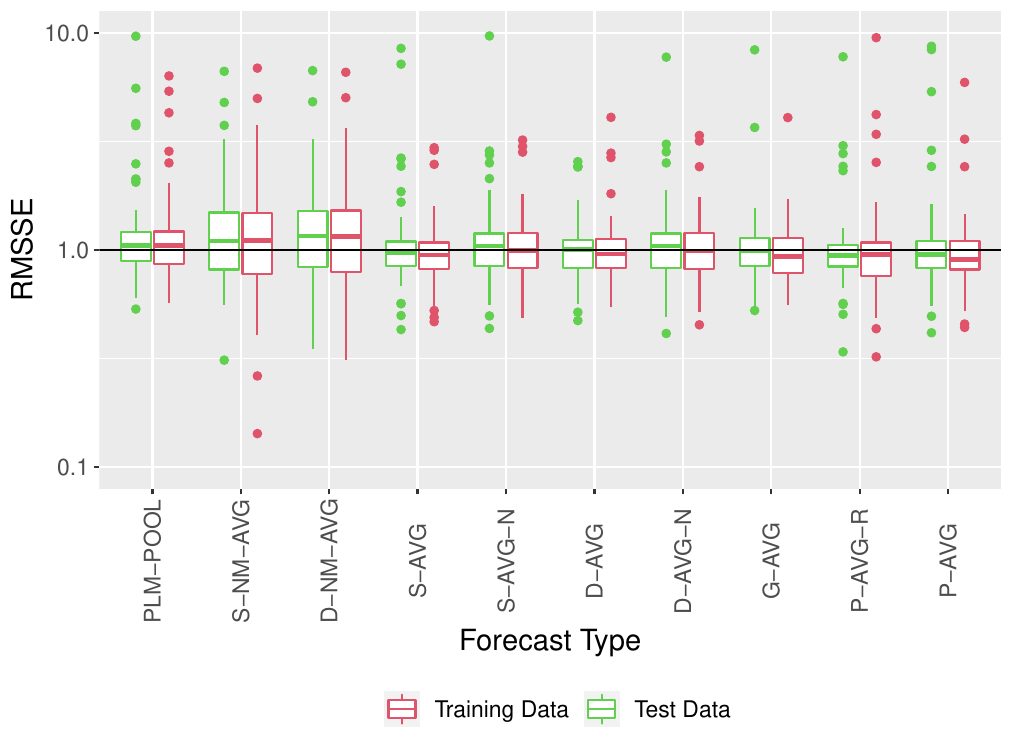}
    \caption{Global RMSSE with respect to ETS forecasts.}
    \label{fig:global_rmsse_f}
\end{figure}

Next, we want to dig deeper into the performance of selected individuals. For this, we selected $4$ of the $43$ individuals to analyze in more detail. Figure \ref{fig:rmsse_f} shows the RMSSE with respect to ETS forecasts for these selected fridges. We also show standard errors of the RMSSE. Due to the different behaviors of the individuals, we also observe different performances of the averaging approaches. While for individual $3$ we can improve forecasts using model-based averaging methods, this is not the case for individual $27$. We also observe that the panel benchmark model PLM-POOL does not always yield better forecasts and is often even surpassed by any of the averaging approaches.

\begin{figure}[!ht]
    \center
    \includegraphics[width=0.8\textwidth]{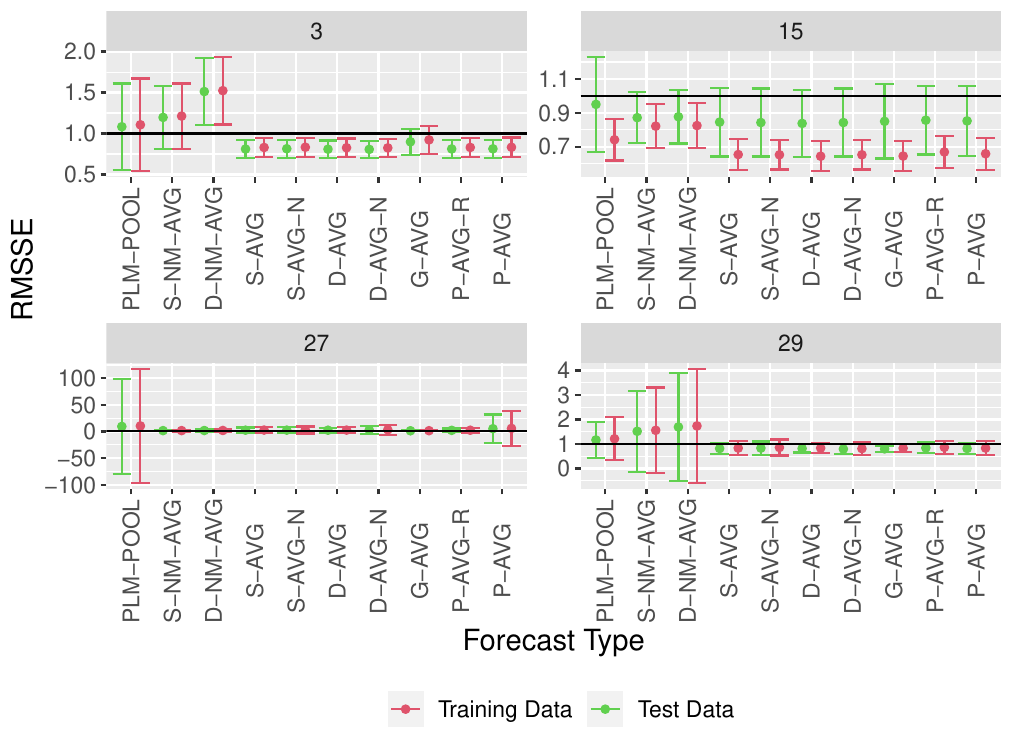}
    \caption{Individual RMSSE with respect to ETS forecasts.}
    \label{fig:rmsse_f}
\end{figure}

We can dig even deeper into evaluating the errors when looking at the evolution of the RMSSE. For reasons of clarity, Figure \ref{fig:run_rmsse} shows the running RMSSE with respect to the RW forecasts for the distance-based averaging methods and the two benchmark methods of ETS and PLM-POOL. The dashed vertical line indicates the split between training and test periods. These plots really show the possibility of improving the benchmark forecasts by smartly taking averages of neighbors' forecasts. Similar plots for the remaining averaging methods can be found in the appendix.

\begin{figure}[!ht]
    \center
    \includegraphics[width=0.8\textwidth]{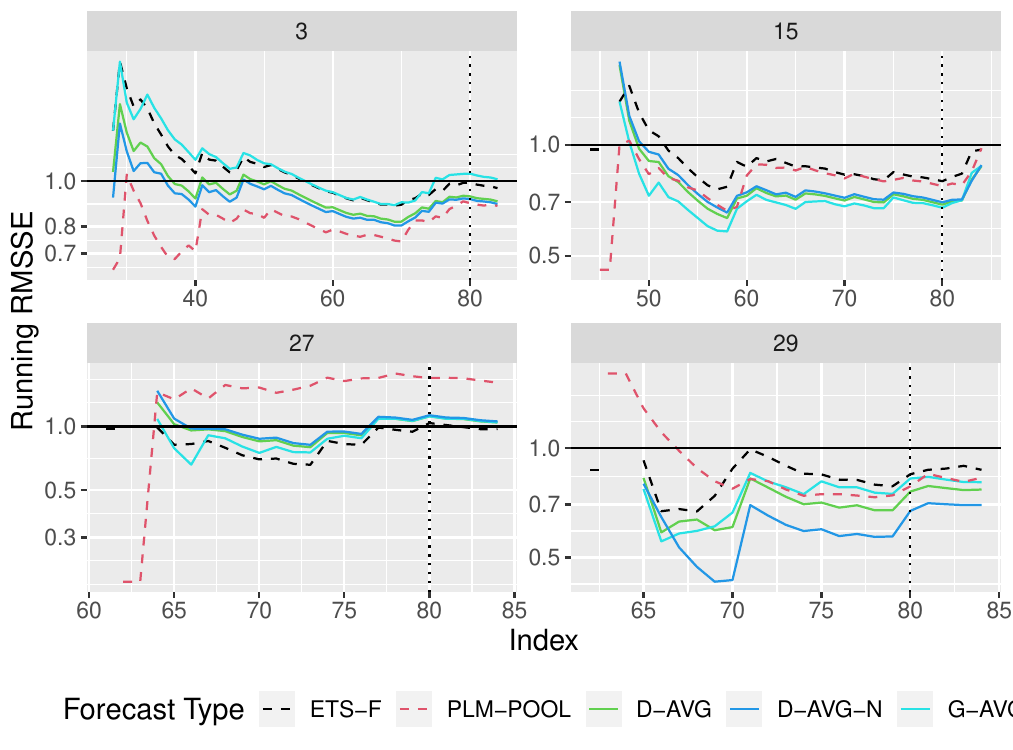}
    \caption{Running RMSSE with respect to random-walk forecasts for distance-based averaging. The dotted lines indicate the split between training and test set.}
    \label{fig:run_rmsse}
\end{figure}

\subsubsection{Further Diagnostics}

To understand better the differences in performance on individual level, we perform further diagnostics. For fixed neighborhood size of at most $5$ neighbors, we first consider the median, minimum and maximum normalized DTW distance  and their evolution over time. This analysis might already give hints on the homogeneity of the neighborhoods (Figure \ref{fig:dtw_evol}). We observe that for individuals $3$, $15$ and $29$ the distances and hence the neighborhoods seem to stabilize after the first few steps while for individual $27$ the normalized distance increases until the end. This means that this individual becomes harder and harder to match. Overall, individual $27$ is too heterogeneous and lacks homogeneity to improve forecasts as also indicated by the high mean $2$-Wasserstein distance in Figure~\ref{fig:wdists}.

\begin{figure}[!ht]
    \center
    \includegraphics[width=0.8\textwidth]{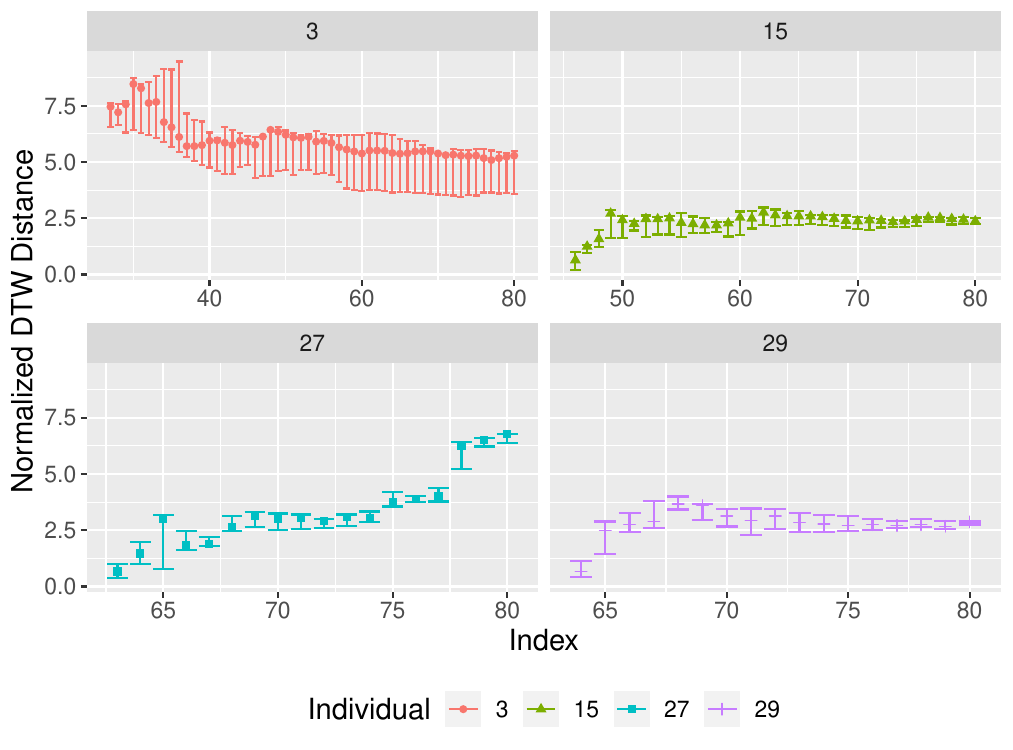}
    \caption{Normalized DTW distance evolution over time.}
    \label{fig:dtw_evol}
\end{figure}

\begin{figure}[!ht]
    \center
    \includegraphics[width=0.8\textwidth]{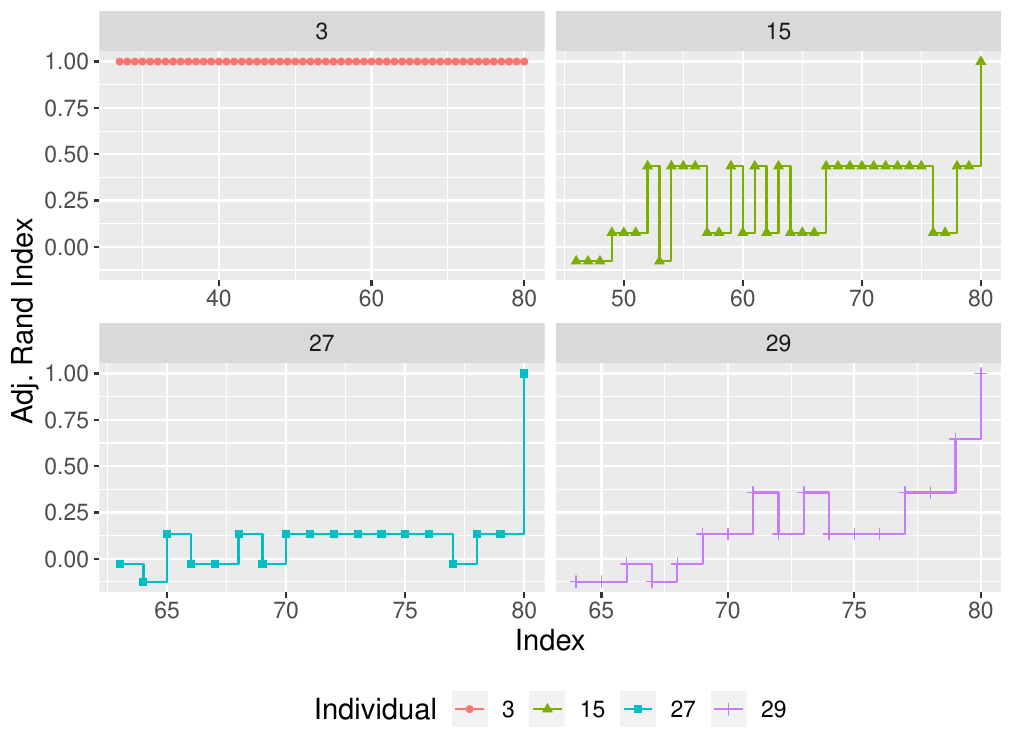}
    \caption{Adjusted Rand index over time.}
    \label{fig:rand_evol}
\end{figure}

Another way to evaluate the neighborhoods is as follows. We consider the final training neighborhoods as the ground truth and compare each neighborhood yielded in the training process to this ground truth using the adjusted Rand index (\citet{randindex}, \citet{adj_rand}). Figure \ref{fig:rand_evol} shows the adjusted Rand index for each time step. We see that individual $3$ is perfectly matched from the beginning which is also not suprising since this individual only has very few possible neighbors. However, there is much more variability present for the remaining smart fridges, and especially for individual $27$ the adjusted Rand index is hardly increasing, indicating difficulties in finding the proper neighbors for this individual.

Finally, we take a look at the $2$-Wasserstein distances as discussed in \ref{sec:th_mot}. Since we only introduced the Wasserstein for a basic $ANN$ model, we want to go into more detail here. The one-step ahead point forecasts which we obtain and use for averaging are the means of the forecast distribution, hence we can compute the Wasserstein distances in a straight-forward way. Namely, we have that

\begin{align}\label{eq:gen_wdist}
    W_2^2(\hat X_{n+1|n},\hat Y_{n+1|n}) &= (\hat x_{n+1|n}-\hat y_{n+1|n})^2 + (\hat\sigma_X-\hat\sigma_Y)^2,
\end{align}
where $X,Y$ are two arbitrary ETS models with realized point forecasts $\hat x,\hat y$ as well as estimated standard deviations $\hat\sigma_X,\hat\sigma_Y$, respectively.

Figure~\ref{fig:wdists} shows the mean $2$-Wasserstein distances as in Eq.~\eqref{eq:gen_wdist} per selected individual. The neighborhoods considered here are the same as in Figure~\ref{fig:dtw_evol}. The dashed lines indicate the overall mean distance. We clearly observe large distances for individual $27$ which is consistent with the rather bad results. Since the DTW distances as in Figure~\ref{fig:dtw_evol} are increasing for this individual, we cannot expect the Wasserstein distances to be small, and thus also cannot guarantee a reduction in forecast error. Indeed, the overall mean distance is $12.5$ which is more than twice as large as the mean values for the remaining individuals ($3.2$ for individuals $3,15$, and around $5.2$ for individual $29$.) selected in the analysis. This result empirically verifies the extension of our motivating theoretical results.

\begin{figure}[!ht]
    \center
    \includegraphics[width=0.8\textwidth]{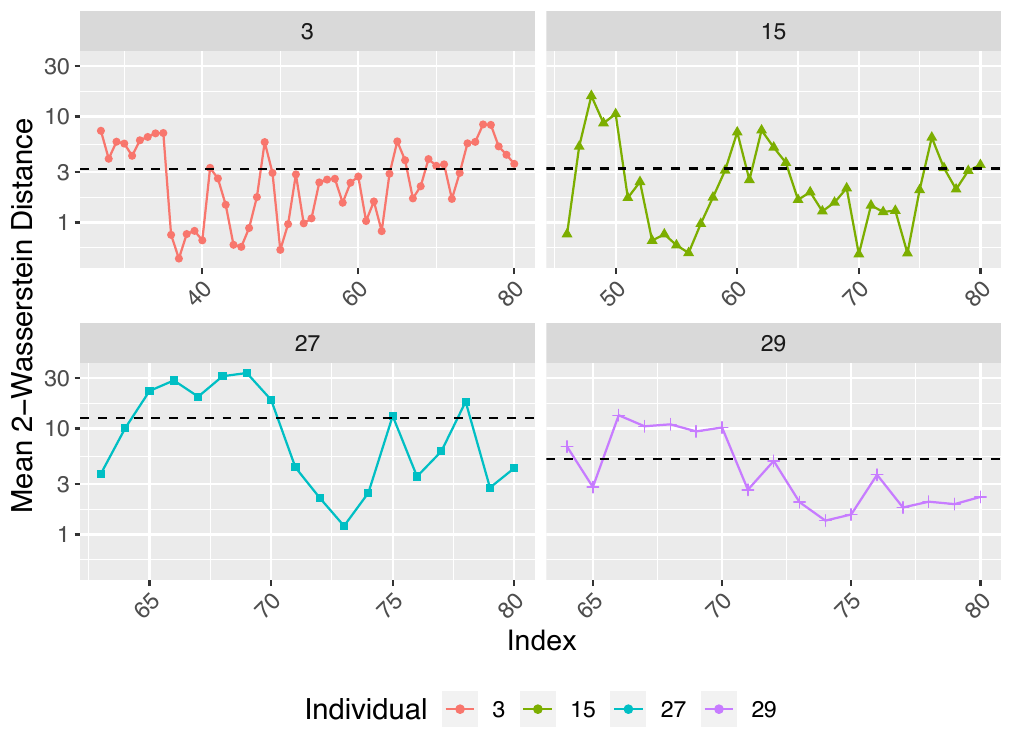}
    \caption{Wasserstein distances of the forecast distributions.}
    \label{fig:wdists}
\end{figure}

\subsubsection{The Best Model}
\input{paper_tables/rmsse_fridges.tex}
By miniming the mean RMSSE values, we can choose the optimal averaging technique on the dataset. With a mean RMSSE value with respect to the ETS base forecast of $1.08$ the optimal averaging method is P-AVG-R. This rather high value is due to bad performance of a small amount of individuals, skewing the mean value. The median values accounts to $0.85$. Table~\ref{tab:tbl:errors_fridges} shows the results in terms of a number of error measures. While for individual $3$ the averaging yields uniformly better results, it is not the case for individual $15$ where the base ETS model and the pooled global model outperform the averaging methodology.

To test statistical significance, we follow the suggestions by \citet{demsar2006statistical}. A non-parametric Friedman test is performed at an $\alpha$-level of $5\%$ for each error measure using the results for all individuals. This yields significant differences for all error measures ($p<0.01$). Post-hoc pairwise Wilcoxon tests with Holm's $\alpha$-level correction are then performed. However, these test does not yield any significant pairwise differences in the methods for any error measure.

Additionally,  we take a look at the distributional properties of the RMSSE values obtained on the test set and compare the methods based on how many individuals could be improved by averaging. First, Figure \ref{fig:rmsse_perc} shows the ordered RMSSE test set values with respect to the ETS forecast for each averaging method. The vertical lines indicate the percentile of individuals with smaller error compared to the ETS benchmark. For an overall useful methodology we want the vertical lines to be greater than $0.5$, meaning we can improve the forecasting for more than $50\%$ of the individuals. This is the case for the global averaging method G-AVG, the performance-based one P-AVG, P-AVG-R, as well as the simple S-AVG, whereby the best improvement is obtained by the weighted average of using errors of the refitted models (P-AVG-R). Furthermore, we see that all methods not including the individual, namely D-AVG-N and S-AVG-N, yield worse results. Also, the non-model based methods which only regard the DTW matchings do not seem to be very useful at all since they only decrease the forecast error compared to the ETS benchmark for around $36\%$, and $38\%$ of the individuals, respectively.

\begin{figure}[!ht]
    \center
    \includegraphics[width=0.8\textwidth]{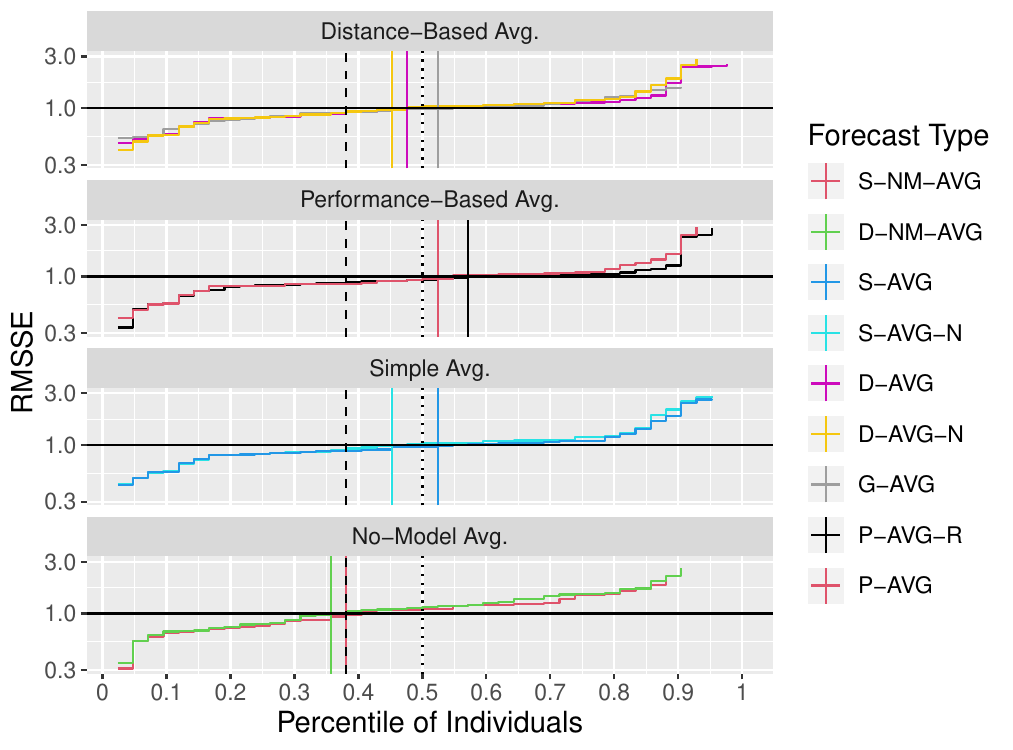}
    \caption{Percentile plots for each averaging method with respect to ETS forecast. The dashed lines show the performance of the global panel benchmark model. The dotted lines indicate the $50\%$ percentile.}
    \label{fig:rmsse_perc}
\end{figure}

\section{Conclusions}\label{sec:conc}

We present a general framework for the improvement of individual forecasts in a set of possibly heterogeneous time series. It is based on a theoretical motivation regarding simple state-space models in terms of exponential smoothing, however, this motivation may be extended to more general models. In this theoretical part of the work, new approaches are proposed by computing the DTW distance explicitly on random processes as well as using the Wasserstein distance to compare the forecast distribution of state-space models. Additionally, dynamic time warping is introduced and described in a far more rigorous way.

We apply a sequence matching procedure in dynamic time warping which allows us to find similar time series in terms of their shape and which works for any two time series. To this end, we emphasize the use of asymmetric matching and the corresponding extension of the global averaging methodology to this type of matching. The transparency of the algorithm also enables us to perform diagnostics and to understand when this procedure is appropriate and yields reasonable results. This point also differentiates our work to machine learning approaches which tend to be intransparent and require a large total number of observations to work efficiently.

\review{
    The models we use for our analysis are ETS models, yet any family of models can be used in our procedure, making it very flexible in practice. It also extends automatic forecasting frameworks such as ARIMA or ETS naturally. While the methodology may not be able to improve performance of the base model all the time, it still yields competitive results. Overall, this framework fits many real world applications where one encounters a set of heterogeneous and possibly short time series.
}

Some aspects are still to be considered. Extension of the work could be using an adaptive data-driven number of neighbors for each time series instead of a fixed one.
Further, a robustification of the procedure could be beneficial since dynamic time warping especially is quite sensitive to outliers in the time series.

%% -- Optional special unnumbered sections -------------------------------------
\newpage
\section*{Computational Details}
The results in this paper were obtained using
\proglang{R}~4.1.3. Figures were produced by \pkg{ggplot2} \citep{ggplot2} and tables were created using \pkg{knitr} \citep{knitr} as well as \pkg{kableExtra} \citep{kableExtra}. \proglang{R} itself
and all packages used are available from the Comprehensive
\proglang{R} Archive Network (CRAN) at
\url{https://CRAN.R-project.org/}.

The source code of this paper in form of the \proglang{R} package \pkg{TSAvg} is available from GitHub at \url{https://github.com/neubluk/TSAvg}.

\section*{Acknowledgments and Disclosure of Funding}

We acknowledge support from the Austrian Research Promotion Agency (FFG), Basisprogramm project ``Meal Demand Forecast'', and Schrankerl GmbH for the cooperation and access to their data.

\clearpage

\appendix
\section{Methodology}
\subsection{Dynamic Time Warping}\label{sec: dtw2}
The set of all allowed warping functions $\Phi$ depends on the type of time-warping which is usually parametrized by a \textit{step pattern}. This pattern defines the allowed values of $\phi(k)$ given $\phi_k(1),\dots,\phi_k(k-1)$ for each $k$.

Basic properties of a warping function are as follows.
\begin{itemize}
    \item \textbf{Monotonicity}. We have $\phi_{\mathbf X}(k)\leq\phi_{\mathbf X}(k+1),\phi_{\mathbf Y}(k)\leq\phi_{\mathbf Y}(k+1)$ for all $k$. This allows for only reasonable matchings where the order of the sequences is unchanged.
    \item \textbf{Slope}.
    The possible (local) slope of $\phi$ is given by
    \begin{align*}
        Q_{k_s} = \cfrac{\sum_{i=1}^{T_{k_s}} q_i^{(k_s)}}{\sum_{i=1}^{T_{k_s}} p_i^{(k_s)}},
    \end{align*}
    where $s$ denotes the step pattern for the warping function, $T_{k_s}$ gives the number of allowed steps in the $k$-th step pattern path, and $p_i^{(k_s)},q_i^{(k_s)}$ gives the actual size of steps in $\mathbf X$ and $\mathbf Y$ direction, respectively. Thus, we have that $p_i^{(k_s)}=\phi_{\mathbf X}(i+1)-\phi_{\mathbf X}(i)$ for any $i$, and $q_i^{(k_s)}$ analogously. 

    The slope is therefore constrained by $\min_{k_s}Q_{k_s}\leq Q_{k_s}\leq\max_{k_s}Q_{k_s}$.
    \item \textbf{Normalizability}. Depending on $\phi$, the DTW distance can also be normalized by $M_\phi=\sum_{k=1}^Kw_\phi(k)$, i.e.~we have $\text{nDTW}(\mathbf X,\mathbf Y):=\text{DTW}(\mathbf X,\mathbf Y)/M_\phi$. 
\end{itemize}

More properties arise when moving to more concrete warping functions.

\subsubsection{Symmetric DTW}

The most common type is the so-called \textit{symmetric} matching pattern where
\begin{itemize}
    \item the endpoints must be matched, meaning we have $\phi(1)=(\phi_\mathbf X(1),\phi_\mathbf Y(1))=(1,1)$ and $\phi(K)=(\phi_\mathbf X(K),\phi_\mathbf Y(K))=(n,m)$,
    \item all points must be matched, i.e. $\phi_{\mathbf X}(k+1)-\phi_{\mathbf X}(k)\leq 1,~\phi_{\mathbf Y}(k+1)-\phi_{\mathbf Y}(k)\leq 1$ (this can also be seen as a continuity constraint), 
    \item allowed steps to get to $(i,j)$ are $(i-1,j)\rightarrow (i,j),~(i,j-1)\rightarrow (i,j),~(i-1,j-1)\rightarrow (i,j)$ (see also Eq. \eqref{dtw:step:sym}), implying that
    \item the slope of $\phi$ is unconstrained allowing for an arbitary amount of time stretching or compression because $Q = q/p,~p,q\in\{0,1\}$ and $0\leq Q\leq\infty$, and
    \item the weights are chosen to be $w_\phi(k)=\phi_{\mathbf X}(k)-\phi_{\mathbf X}(k-1) + \phi_{\mathbf Y}(k)-\phi_{\mathbf Y}(k-1)$ such that $M_\phi = n+m$.
\end{itemize}

In fact, the step pattern can be written using a \textit{Dynamic Programming} \citep{dynamic_prog} approach by

\begin{align}\label{dtw:step:sym}
    g(i,j)=\min\left(
        \begin{matrix}
        g(i,j-1)+d(i,j)\\
        g(i-1,j-1)+2d(i,j)\\
        g(i-1,j)+d(i,j)
    \end{matrix}
    \right),\quad 1\leq i\leq n,1\leq j\leq m,
\end{align}
such that $g(1,1)=w(1)d(1,1)=2d(1,1)$ and $g(n,m)=:\text{DTW}_\text{sym}(\mathbf X,\mathbf Y)$, where $n=|\mathbf X|,m=|\mathbf Y|$. The resulting matrix $G$ of Eq.~\eqref{dtw:step:sym} is called the \textit{warping matrix}. The corresponding warping path can be extracted from these recursive calculations by backtracking.

This distance is naturally symmetric and positive definite. However, it is not a metric since the triangle inequality is usually not fulfilled.

This pattern type can be useful when considering sequences of similar lengths. However, for differently sized sequences this might not work reasonably anymore due to above constraints. 

\subsubsection{Asymmetric DTW}

We are particularly interested in subsequence matching where we relax the endpoint constraints. This is also denoted as \textit{open-begin open-end} (OBE) matching \citep{dtwR}. The corresponding optimization problem looks as follows.
$$
\begin{aligned}
    \text{DTW}_{\text{OBE}} (\mathbf X,\mathbf Y) := \min_{1\leq p\leq q\leq m} \text{DTW}(\mathbf X,\mathbf Y^{(p,q)}),
\end{aligned}
$$
where $\mathbf Y^{(p,q)}=(\mathbf Y_p,\dots,\mathbf Y_q)'\in\mathbb R^{(q-p+1)\times d}$ is a subsequence of $\mathbf Y$, and $n\leq m$. This type of matchings requires an \textit{asymmetric} matching pattern. The most basic one is given by
\begin{align}\label{eq:asym_dtw}
    g(i,j)=\min\left(
        \begin{matrix}
        g(i-1,j)+d(i,j)\\
        g(i-1,j-1)+d(i,j)\\
        g(i-1,j-2)+d(i,j)
    \end{matrix}
    \right),\quad 1\leq i\leq n, p\leq j\leq q,
\end{align}
implying that continuity is not imposed anymore. In fact, information of $\mathbf Y$ can be skipped. The slope is constrained by $\min Q=0,~\max Q=2$. Further, the distance is only normalizable by $n$ since we consider subsequences of $\mathbf Y$ of different lengths. Because of the subsequence matching, we have an initial value of $g(1,p)=d(1,p)$ and $g(n,q)=:\text{DTW}_{\text{asym, OBE}} (\mathbf X,\mathbf Y)$.

\subsubsection{Other Step Patterns}

Many more step patterns exist with different types of slope constraints and other properties. Unfortunately, it is often not very clear which step pattern really is the most appropriate one to use. A general family of step patterns is defined in \citet{dtw} by the notion of \textit{symmetricP$x$}, and \textit{asymmetricP$x$}, respectively, where $x$ controls the slope parameter. Usual values are $0$, $0.5$, $1$ and $2$. The larger the slope can possibly be, the more complicated the dynamic programming equations tend to get. As a limiting step pattern one obtains the \textit{rigid} step pattern for $x\rightarrow\infty$. As suggested in \citet{dtwR}, this step pattern is only reasonable when considering an OBE matching, since it finds the most appropriate subsequence without any gaps and, in fact, does not perform any time warping.

To summarize, DTW is a very flexible sequence matching method which yields the optimal matching as well as an associated distance between any two sequences.

\subsection{DTW Averaging}\label{sec: dtw_avg}

As all dissimilarity measures, DTW can also be used to find a barycenter of a set of sequences $\mathcal S$. A barycenter is usually defined to minimize the sum of distances in the set. Generally, $c\in X$ for some metric space $(X,\delta)$ is called a barycenter of $Y\subset X$ if
\begin{align}
    \label{barycenter}
    \sum_{x\in Y}\delta(c,x)\leq\sum_{x\in Y}\delta(y,x),
\end{align}
for any $y\in X$. 

In terms of time series and DTW we have $Y=\mathcal T$ as the finite set of time series of interest, $X\supset Y$ as the set of all possible time series with some maximum length, and $\delta=\text{DTW}$. However, the space of possible solutions is very large which makes the search for the average difficult. Therefore, approximative solutions have been developed.

In the context of DTW, it is not straightforward to provide a definition for an average. For the symmetric case of DTW, \citet{dtw_pw_avg} have considered pairwise, coordinate-wise averaging, where two sequences are averaged to one sequence until only one sequence is left. This method is easy to implement, but a big downside is that it heavily depends on the order of sequences. A different approach, introduced in \citet{dba}, is global averaging. Starting from an initial sequence, the average is updated in each iteration based on all matchings of all sequences in the set. As mentioned in their paper, this heuristic naturally reduces the sum %in \eqref{barycenter}.
of the warping distances to the average in Eq.~\eqref{barycenter}.

One aspect to consider is the length of the averaging sequence. In pairwise averaging, this average can grow twice as long in each step. In global averaging, the length of the resulting average is fixed to the length of the initial sequence.

We adapt the global averaging methodology to the asymmetric DTW use case as follows. Denote $\mathcal T$ a set of time series of different lengths. Our aim is to have a longer barycenter than the time series of interest, hence as the initial average time series we take the longest one available, i.e. $\mathbf C_0=\text{argmax}_{\mathbf X\in\mathcal T}|\mathbf X|$. Then iteratively we compute $\mathbf C_i,~i=1,2,\dots$: 
\begin{itemize}
    \item Compute $\text{DTW}_\text{asym, OBE}(\mathbf X,\mathbf C_i)$ for all $\mathbf X\in\mathcal T$ to obtain $\phi^{(\mathbf X)}(k):=(\phi_{\mathbf X}(k),\phi_{\mathbf C_i}(k))$ for $k=1,\dots,K_\phi$.
    \item For each time step $t=1,\dots,|\mathbf C_0|$ of $\mathbf C_i$, denoted by $\mathbf C_{i,t}$, let $\mathbb C_{i,t}:=\left\{\mathbf X_j:\mathbf X\in\mathcal T, \phi^{(\mathbf X)}=(j,t)\right\}$ be the set of all associated $\mathbf X$ time steps.
    \item Update each time step of the averaging sequence by
    \begin{align*}
        \mathbf C_{i+1,t} = 
        \begin{cases}
            \cfrac{1}{|\mathbb C_{i,t}|}\sum_{\mathbf Z\in\mathbb C_{i,t}}\mathbf Z & \text{if}\quad |\mathbb C_{i,t}|>0,\\
            \mathbf C_{i,t} &\text{otherwise.}
        \end{cases}
    \end{align*}
\end{itemize}

We perform this iteration for a fixed number of $I$ times or until the sum in Eq. \eqref{barycenter} does not decrease anymore. We denote now the averaging time series as $\text{aDBA}(\mathcal T):=\mathbf C_{\min(S,I)}$ where $S>S^\ast$ is the time of no further reduction after a start-up period $S^\ast$ and the averaging function $\text{aDBA}: \mathcal T\rightarrow \mathcal Z$ where $\mathcal Z$ denotes the set of all possible time series.

\subsection{Theoretical Motivation}\label{sec:th_mot}

We want to motivate our approach of using the DTW distance. For that we consider just a simple case. Let $X$ be an $ANN$ model \citep{expsmoothing_book}. An $ANN$ model is also known as \textit{Exponential Smoothing}. It does not have any trend or seasonality component, and for a time series $(X_t,t=1,\dots,n)$ it is given by the recursion
\begin{align}\label{eq:ann}
    l_t^X &= \alpha X_t + (1-\alpha)l_{t-1}^X,\\
    \hat X_{t+h|t} &= l_t^X,\quad h>0,\nonumber
\end{align}
where $l^X$ denotes the level component of the model which is also equal to the flat forecast $\hat X_{t+h|t}$ for any $h>0$. The model parameter $\alpha$ is usually found by minimizing the sum of squared forecast errors or by maximum likelihood.

\subsubsection{Theoretical DTW Computation}

Let $X,Y$ be two independent $ANN$ models. Without loss of generalization, we assume both initial values are equal to $0$, i.e. $l_0^X=l_0^Y=0$ (otherwise we could just look at $X-l_0^X$). We write $X=(X_1,\dots,X_n)'\sim ANN(\alpha_X,\sigma^2_X)$ and $Y=(Y_1,\dots,Y_n)'\sim ANN(\alpha_Y,\sigma^2_Y)$. For both we consider an equal length of $n=|X|=|Y|$. Considering Eq. \eqref{eq:ann}, we can rewrite the recursive equation using a state-space representation as
\begin{alignat*}{2}
    X_t&= l_{t-1}^X+\epsilon_t, & \quad Y_t&= l_{t-1}^Y+\eta_t\\
    l_t^X&= l_{t-1}^X+\alpha_X\epsilon_t, & \quad l_t^Y&= l_{t-1}^Y+\alpha_Y\eta_t,
\end{alignat*}
where the innovations $\epsilon_t\overset{iid}{\sim}N(0,\sigma^2_X),\eta_t\overset{iid}{\sim}N(0,\sigma^2_Y)$ are assumed to be also pairwise independent for $t=1,\dots,n$. The innovations are the one-step ahead forecast errors given by $\epsilon_t=X_t-l_{t-1}^X$. The states $l^X$ are latent and only $X_t$ itself is observable.

Next, we give an explicit expression for the asymmetric DTW distance between $X$ and $Y$ when considering the squared $L^2$ cross-distance. We denote $X$ to be independent of $Y$ if and only if $X_i$ is independent of $Y_j$ for all $i,j$. This is a simple assumption, however, following results can be easily extended to a more general setup.
\begin{lemma}
    Let $X\sim ANN(\alpha_X,\sigma^2_X),~Y\sim ANN(\alpha_Y,\sigma^2_Y)$ be two independent and centered Exponential Smoothing processes of length $n$. Let $d(i,j):=\mathbb E[(X_i-Y_j)^2]$ denote the cross-distance between $X$ and $Y$.\\
    Then the asymmetric DTW distance is given by
    \begin{align}\label{eq:dtw_ann}
        \text{DTW}_\text{asym}(X,Y)=\sigma^2_X\left(n+ {\binom{n}{2}} \alpha_X^2\right)+\sigma^2_Y \left(n+\left\lfloor \frac{n^2}{4}\right\rfloor\alpha_Y^2\right).
    \end{align}
\end{lemma}
\begin{proof}
    First, note that the cross-distance is given by
    \begin{align*}
        d(i,j) &= \mathbb E[X_i^2]+\mathbb E[Y_j^2]\\
        &= \sigma_X^2(1+(i-1)\alpha_X^2) + \sigma_Y^2(1+(j-1)\alpha_Y^2),
    \end{align*}

    due to the independence of $X_i,~Y_j$ for every $i,~j$. Next, we need to recursively compute the warping matrix $G$ as in Eq. \eqref{eq:asym_dtw}. We have $g(1,1)=d(1,1)=\sigma_X^2+\sigma_Y^2$. Due to the cross-distance being increasing in both $i$ and $j$, we can easily solve Eq. \eqref{eq:asym_dtw} by
    \begin{align*}
        g(i,j) &= \min\begin{pmatrix}
            g(i-1,j)\\
            g(i-1,j-1)\\
            g(i-1,j-2)
        \end{pmatrix}
        +d(i,j)\\
        &=g\left(\tilde i,\tilde j\right) + \sum_{k=0}^{i-\tilde i-1}d(i-k,j-2k),
    \end{align*}
    where $\tilde i,~\tilde j$ denotes the indices where the recursion needs more specific computation. We need this distinction because for small $i,~j$ the minimum value is different than expressed in the sum due to some of the values being not assigned. In detail, we have that
    \begin{align*}
        g\left(\tilde i,\tilde j\right)=\begin{cases}
            \sum_{k=0}^{\tilde i - 1}d(\tilde i-k,1) & \text{if }\tilde i>0, \tilde j=1,\\
            g(\tilde i-1,1)+d(\tilde i, 2) & \text{if }\tilde i>1,\tilde j=2,\\
            \text{NA}&\text{if }\tilde i=1,\tilde j>1.\\
        \end{cases}
    \end{align*}
    Altogether, we obtain, by induction,
    \begin{align}\label{eq:dtw_ann_proof}
        g(i,j) = \begin{cases}
            i(\sigma_X^2+\sigma_Y^2)+ {\binom{i}{2}}\sigma_X^2\alpha_X^2+\lfloor \frac{j^2}{4}\rfloor\sigma_Y^2\alpha_Y^2 & \text{if }2i-j\geq 1,\\
            \text{NA} & \text{otherwise.}
        \end{cases}
    \end{align}

    Setting $i=j=n$ finishes the proof.
\end{proof}

\subsubsection{Relation to Wasserstein Distance}

Since we are interested in the one-step ahead forecast, we may look at the corresponding forecast distributions. The forecast distributions are given by the conditional distributions of $\hat X_{n+1|n} = X_{n+1}|l_n^X\sim N(l_n^X,\sigma_X^2)$ and $\hat Y_{n+1|n}=Y_{n+1}|l_n^Y\sim N(l_n^Y,\sigma_Y^2)$, respectively. Now we want to measure the distance between those two distributions. For that we use the $2$-Wasserstein distance, which is defined as follows \citep{Villani2009}. Let $\mu,\nu$ be two measures and $\pi(\mu,\nu)$ be the set of all couplings of $\mu$ and $\nu$. Then
\begin{align*}
    W_2(\mu,\nu) = \sqrt{\inf_{\gamma\in\pi(\mu,\nu)}\int ||x-y||^2d\gamma(x,y)}.
\end{align*}

In the case of random variables we can also write 
\begin{align*}
    W_2(X,Y)=\sqrt{\inf\{\mathbb E[||\tilde X-\tilde Y||^2]: (\tilde X,\tilde Y)\in\pi(X,Y) \}},
\end{align*}
for any two random variables $X,Y$. If both are Gaussian, i.e. $X\sim N(\mu_1,\Sigma_1)$ and $Y\sim N(\mu_2,\Sigma_2)$, then the squared $2$-Wasserstein distance is readily computed by
\begin{align}\label{eq:gaus_wdist}
    W_2^2(X,Y) = ||m_1-m_2||^2 + \text{tr}\left(\Sigma_1+\Sigma_2-2(\Sigma_1^{1/2}\Sigma_2\Sigma_1^{1/2})^{1/2}\right).
\end{align}

Details are available in the work of \citet{gaussian_wasserstein}. Applying Eq. \eqref{eq:gaus_wdist} to the forecast distributions of the $ANN$ models, we can calculate the $2$-Wasserstein distance between $\hat X_{n+1|n}$ and $\hat Y_{n+1|n}$ yielding

\begin{align}\label{eq:w2}
    W_2^2(\hat X_{n+1|n},\hat Y_{n+1|n}) &= (l_n^X-l_n^Y)^2 + (\sigma_X-\sigma_Y)^2
\end{align}

Thus, both Eq. \eqref{eq:dtw_ann}, and \eqref{eq:w2} are quadratic in the model parameters of $ANN$ allowing us to give following theorem.

\begin{theorem}\label{th:wasserstein}
    Let $\mathcal X$ be the space of independent $ANN$ processes of length $n>n(p)$ for some $p\in(0,1)$, equipped with the asymmetric DTW distance, and $\mathcal Y$ be the space of corresponding Gaussian forecast distributions equipped with the squared $2$-Wasserstein distance.\\
    Then the map $\mathcal X\rightarrow\mathcal Y:X\mapsto \hat X_{n+1|n}$ is Lipschitz-continuous with Lipschitz constant $L<1$ and probability at least $p$.\
\end{theorem}
\begin{proof}
    Let $X,Y\in\mathcal X$ be two arbitrary $ANN$ processes. We have by Eq.~\eqref{eq:w2} that $W^2_2(\hat X_{n+1|n},\hat Y_{n+1|n})=(l_n^X-l_n^Y)^2+(\sigma_X-\sigma_Y)^2$ for fixed states $l_n^X,~l_n^Y$. Since $l_n^X,~l_n^Y$ are both realizations of independent Gaussian random variables, we obtain 
    \begin{align}\label{eq:w2_chi2}
        \mathbb P\left( (l_n^X-l_n^Y)^2\leq q_p n(\alpha_X^2\sigma_X^2+\alpha_Y^2\sigma_Y^2)\right)=p,
    \end{align}
    where $q_p$ is the $p$-quantile of a $\chi^2(1)$ distribution. Then using Eq.~\eqref{eq:w2_chi2} yields
    \begin{align*}
        p=\mathbb P&\left(\cfrac{W^2_2(\hat X_{n+1|n},\hat Y_{n+1|n})}{\text{DTW}_\text{asym}(X,Y)}\leq \cfrac{q_pn(\sigma_X^2\alpha_X^2+\sigma_Y^2\alpha_Y^2)+(\sigma_X-\sigma_Y)^2}{\sigma^2_X\left(n+ {\binom{n}{2}} \alpha_X^2\right)+\sigma^2_X \left(n+\lfloor \frac{n^2}{4}\rfloor\alpha_Y^2\right)}\right) \leq\\
        \mathbb P&\left(\cfrac{W^2_2(\hat X_{n+1|n},\hat Y_{n+1|n})}{\text{DTW}_\text{asym}(X,Y)}\leq \cfrac{q_p(\sigma_X^2(1+n\alpha_X^2)+\sigma_Y^2(1+n\alpha_Y^2))}{\sigma^2_X\left(n+ {\binom{n}{2}} \alpha_X^2\right)+\sigma^2_X \left(n+\lfloor \frac{n^2}{4}\rfloor\alpha_Y^2\right)}\right)\leq\\
        \mathbb P&\left(\cfrac{W^2_2(\hat X_{n+1|n},\hat Y_{n+1|n})}{\text{DTW}_\text{asym}(X,Y)}<1\right)
    \end{align*}
    using that ${\binom{n}{2}},\lfloor n^2/4\rfloor > n$ for $n>n(p)$. Thus, the map $X\mapsto \hat X_{n+1|n}$ is Lipschitz-continuous with constant $L<1$ and probability at least $p$.
\end{proof}
\begin{remark}
    If we want to have Lipschitz-continuity with at least $95\%$ probability, then $q_{0.95}<4$ and Theorem~\ref{th:wasserstein} holds with $n>16$.
\end{remark}
\begin{remark}
    This result also holds when looking at the normalized DTW measure with Lipschitz constant $L>1$ and is therefore not a contraction anymore.
\end{remark}

\begin{remark}
    It also tells us that close time series in terms of DTW are also close in their corresponding forecast distribution both in mean and variance. Further, it assures us that small changes in the time series only affect the difference in forecast distributions by a small amount.
\end{remark}

More detailed results are obtained when considering the mean forecasts given by $\mathbb E[X_{n+1|n}]:=\mathbb E[X_{n+1}|l_n^X]=l_n^X\sim N(0,n\sigma_X^2\alpha_X^2)$, and $\mathbb E[Y_{n+1|n}]=l_n^Y\sim N(0,n\sigma_Y^2\alpha_Y^2)$, respectively. The corresponding $2$-Wasserstein distance is computed to be
\begin{align*}%\label{eq:w2}
    W_2^2(\mathbb E [X_{n+1|n}],\mathbb E [Y_{n+1|n}]) &= n\sigma_X^2\alpha_X^2+n\sigma_Y^2\alpha_Y^2-2\sqrt{n^2\sigma_X^2\sigma_Y^2\alpha_X^2\alpha_Y^2} \nonumber \\
    &= n(\sigma_Y\alpha_X-\sigma_Y\alpha_Y)^2.
\end{align*}

\begin{theorem}
    Let $\mathcal X$ be the space of independent $ANN$ processes of length $n>5$ equipped with the asymmetric DTW distance, and $\mathcal Y$ be the space of corresponding Gaussian forecast distributions equipped with the squared $2$-Wasserstein distance.\\
    Then the map $\mathcal X\rightarrow\mathcal Y:X\mapsto \mathbb E[X_{n+1|n}]$ is Lipschitz-continuous with Lipschitz constant $L<1$.
\end{theorem}
\begin{proof}
    Let $X,Y\in\mathcal X$ be two arbitrary $ANN$ processes. We have that
    \begin{align*}
        \cfrac{W^2_2(\mathbb E[X_{n+1|n}],\mathbb E[Y_{n+1|n}])}{\text{DTW}_\text{asym}(X,Y)}&=
        \cfrac{n(\sigma_X\alpha_X-\sigma_Y\alpha_Y)^2}{\sigma^2_X\left(n+ {\binom{n}{2}} \alpha_X^2\right)+\sigma^2_X \left(n+\lfloor \frac{n^2}{4}\rfloor\alpha_Y^2\right)}\\
        &< 1,
    \end{align*}
    using that ${\binom{n}{2}},\lfloor n^2/4\rfloor > n$ for $n>5$. Thus, the map $X\mapsto \mathbb E[X_{n+1|n}]$ is Lipschitz-continuous with constant $L<1$.
\end{proof}

\subsubsection{Reduction of Mean Squared Error}

Another result is about the relation of DTW and the mean squared error of a convex combination of the mean forecasts. Let $Z_{n+1|n}(w):=w\mathbb E[X_{n+1|n}]+(1-w)\mathbb E[Y_{n+1|n}]$ for $w\in[0,1]$. We have $Z_{n+1|n}(w)=w l_n^X+(1-w)l_n^Y$. However, in practice, the states $l_n$ are not known and need to be estimated by actually estimating the smoothing parameter $\alpha$. To this end, assume there exists unbiased estimators $\hat\alpha_X,~\hat\alpha_Y$ for $\alpha_X,~\alpha_Y$, respectively. We further assume they have finite second moment. The corresponding estimating forecasts are given by $\hat z(w)$. Then, under certain conditions for the estimation errors made for $\alpha_X,~\alpha_Y$ we have following result.

\begin{theorem}
    Let $X\sim ANN(\alpha_X,\sigma^2_X),~Y\sim ANN(\alpha_Y,\sigma^2_Y)$ be two independent and centered Exponential Smoothing processes of length $n$ and known variances $\sigma_X^2,~\sigma_Y^2$.
    Then a convex combination of the forecasts $\hat Z$ reduces the mean squared error, i.e.
    \begin{align*}
        \mathbb E[(X_{n+1}-\hat z(w))^2]\leq\mathbb E[(X_{n+1}-\hat l_n^X)^2],
    \end{align*}
    if $\text{MSE}(\hat\alpha_Y)\leq \frac{\sigma_X^2}{2\sigma_Y^2}\big((1-w^2)\text{MSE}(\hat\alpha_X)-\text{nDTW}_\text{asym}(X,Y)/\sigma_X^2)$.
\end{theorem}
\begin{proof}
    We have that
    \begin{align*}
        \mathbb E[(X_{n+1}-\hat z(w))^2] &= \mathbb E[(l_n^X+\epsilon_{n+1}-(w\hat l_n^X+(1-w)\hat l_n^Y))^2] \\
        &= \sigma_X^2 + \mathbb E[(l_n^X-(w\hat l_n^X+(1-w)\hat l_n^Y))^2],
    \end{align*}
    using the independence of the error term $\epsilon_{n+1}$. Further, we obtain
    \begin{align*}
        \mathbb E[(X_{n+1}-\hat z(w))^2] &\leq \mathbb E[(l_n^X-l_n^Y)^2] + w^2\mathbb E[(\hat l_n^X-\hat l_n^Y)^2] + \mathbb E[(l_n^Y-\hat l_n^Y)^2] \\
        &= n(\sigma_X^2\alpha_X^2+\sigma_Y^2\alpha_Y^2)+w^2(\mathbb E[(\hat l_n^X)^2]+\mathbb E[(\hat l_n^Y)^2])+ \mathbb E[(l_n^Y-\hat l_n^Y)^2].
    \end{align*}
    We can quickly compute the last terms of above by
    \begin{align*}
        \mathbb E[(l_n^Y-\hat l_n^Y)^2] &= \int\mathbb E[ (l_n^Y-\hat l_n^Y)^2|\hat\alpha_Y=a]\mathbb P(\hat\alpha_Y=a)da\\
        &=n\sigma_Y^2 \int (a-\alpha_Y)^2 \mathbb P(\hat\alpha_Y=a)da \\
        &=n\sigma_Y^2 \text{MSE}(\hat\alpha_Y),\quad\text{and}\\
        \mathbb E[(\hat l_n^X)^2] &= n\sigma_X^2\mathbb E[\hat\alpha_X^2]\\
        &= n\sigma_X^2(\text{MSE}(\hat\alpha_X)+\alpha_X^2).
    \end{align*}
    In total we get that
    \begin{align*}
        \mathbb E[(X_{n+1}-\hat z(w))^2] &\leq n\sigma_X^2\big( (1+w^2)\alpha_X^2+w^2\text{MSE}(\hat\alpha_X)\big)+\\
        &\quad n\sigma_Y^2\big((1+w^2)\alpha_Y^2+2\text{MSE}(\hat\alpha_Y)\big)\\
        &\overset{!}{\leq} \mathbb E[(X_{n+1}-\hat l_n^X)^2]\\
        &= n\sigma_X^2\text{MSE}(\hat\alpha_X).
    \end{align*}
    Using that $n\sigma_X^2(1+w^2)\alpha_X^2+n\sigma_Y^2(1+w^2)\alpha_Y^2\leq \text{DTW}_\text{asym}(X,Y)$, this finally yields
    \begin{align*}
        \mathbb E[(X_{n+1}-\hat z(w))^2] &\leq \text{DTW}_\text{asym}(X,Y)+n(w^2\sigma_X^2\text{MSE}(\hat\alpha_X)+2\sigma_Y^2\text{MSE}(\hat\alpha_Y))\\
        &\leq n\sigma_X^2\text{MSE}(\hat\alpha_X),
    \end{align*}
    if $\text{MSE}(\hat\alpha_Y)\leq \frac{\sigma_X^2}{2\sigma_Y^2}\big((1-w^2)\text{MSE}(\hat\alpha_X)-\text{nDTW}_\text{asym}(X,Y)/\sigma_X^2)$.
\end{proof}

\begin{remark}
    In practice, the previous theorems tell us that if $X,Y$ are close in terms of DTW and, additionally, the estimation error made when estimating $\alpha_Y$ is smaller than the error made for $\alpha_X$, then the convex combination forecast improves the point forecast for $X_{n+1}$. 

    The condition of $\text{MSE}(\hat\alpha_Y)<<\text{MSE}(\hat\alpha_X)$ might have various reasons. In an application, the fit of the model might be better for $Y$ than for $X$, resulting in better estimation of $\alpha$.
\end{remark}

\subsubsection{Conclusions}

These theoretical results give us arguments of our methodology for the most simple cases of models. However, the arguments might be extended to a broader family of models as given in \citet{expsmoothing_book}. In practice, we also need to use \textit{open begin, open end} matching since the asymmetric DTW measure cannot be computed once the reference time series is too long. Still similar arguments should hold, and motivate our approach to using DTW neighborhoods and perform model averaging. The theory does not give us hints how to choose the optimal weights, hence we propose the weights of Section~\ref{sec:mavg}.

\section{Further TSCV Evaluation Plots}
\begin{figure}[!ht]
    \center
    \includegraphics[width=0.75\textwidth]{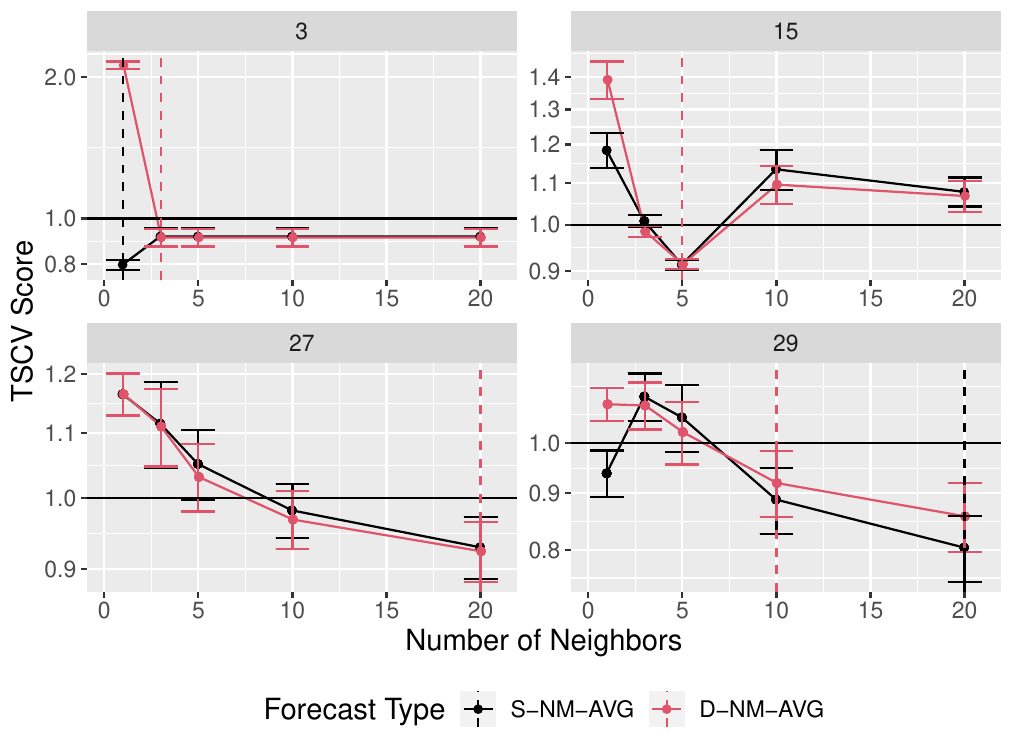}
    \caption{TSCV scores for simple no-model averaging. The dashed vertical lines indicate the optimal number of neighbors.}
\end{figure}
\begin{figure}[!ht]
    \center  
    \includegraphics[width=0.75\textwidth]{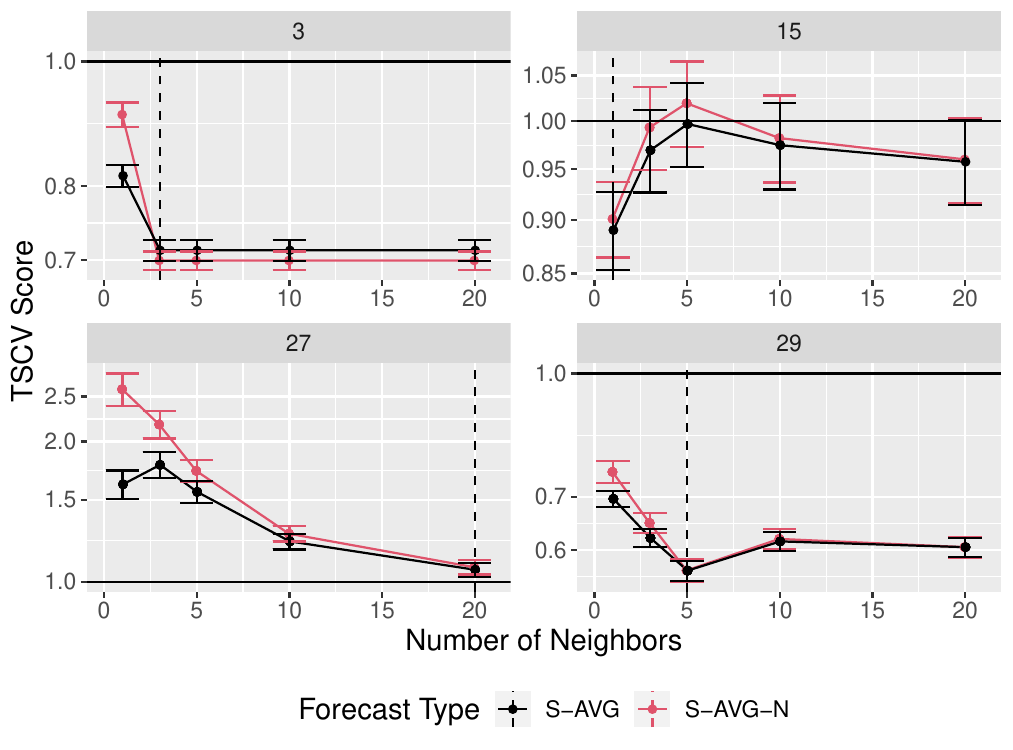}
    \caption{TSCV scores for simple averaging. The dashed vertical lines indicate the optimal number of neighbors.}
\end{figure}
\begin{figure}[!ht]
    \center
    \includegraphics[width=0.75\textwidth]{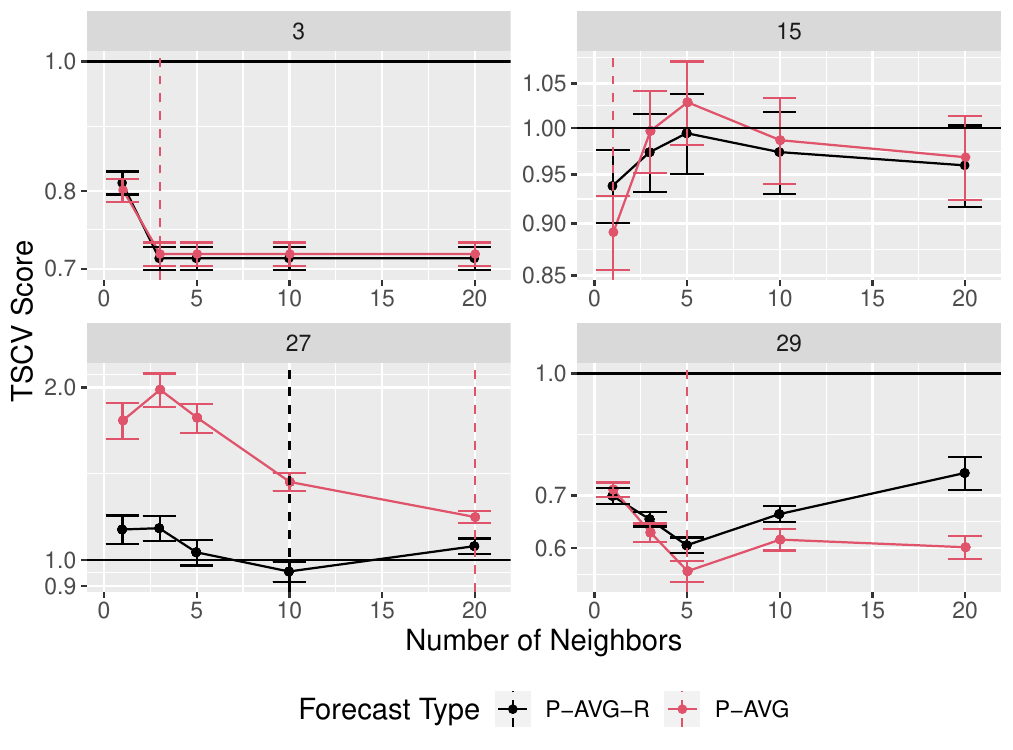}
    \caption{TSCV scores for performance-based averaging. The dashed vertical lines indicate the optimal number of neighbors.}
\end{figure}
\clearpage

\section{Further Evaluation Plots}

\begin{figure}[!ht]
    \center
    \includegraphics[width=0.8\textwidth]{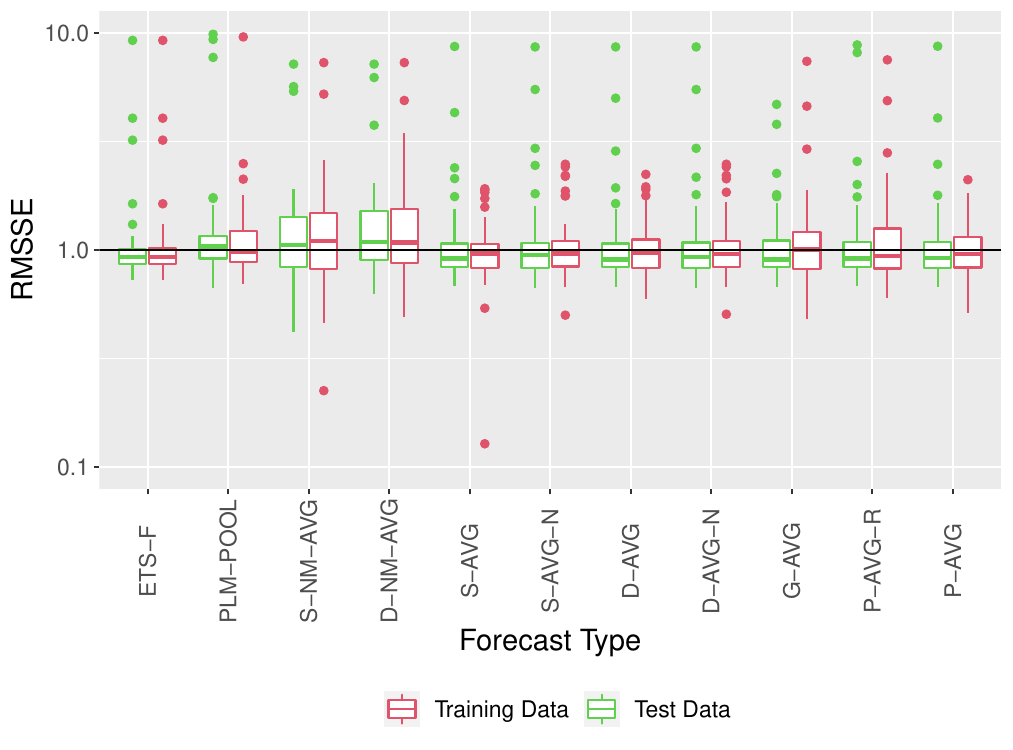}
    \caption{Global RMSSE with respect to random-walk forecasts.}
    \label{fig:global_rmsse}
\end{figure}

\begin{figure}[!ht]
    \center
    \includegraphics[width=0.8\textwidth]{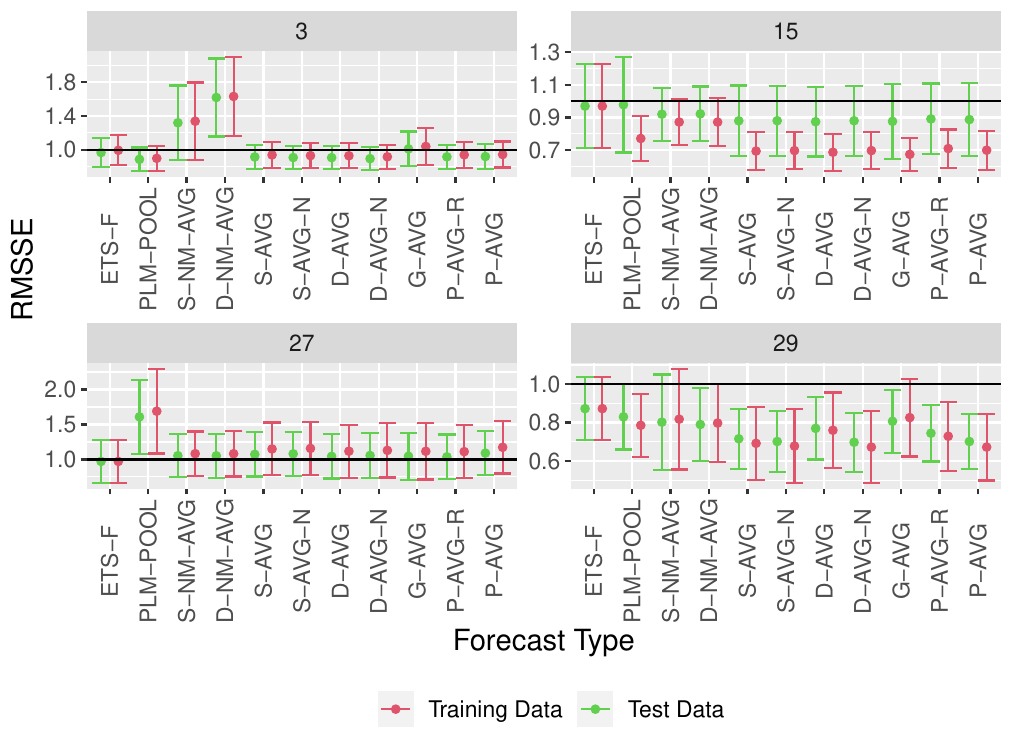}
    \caption{Individual RMSSE with respect to random-walk forecasts.}
    \label{fig:rmsse}
\end{figure}

\begin{figure}[!ht]
    \center
    \includegraphics[width=0.8\textwidth]{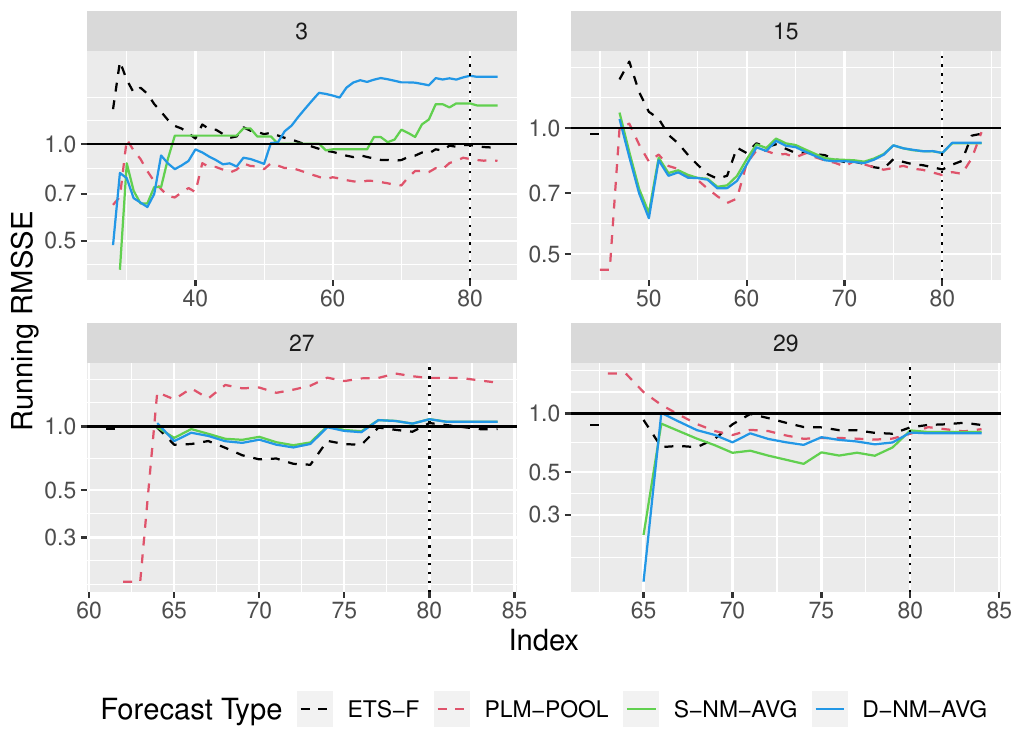}
    \caption{Running RMSSE with respect to random-walk forecasts for simple no-model averaging.}
\end{figure}

\begin{figure}[!ht]
    \center
    \includegraphics[width=0.8\textwidth]{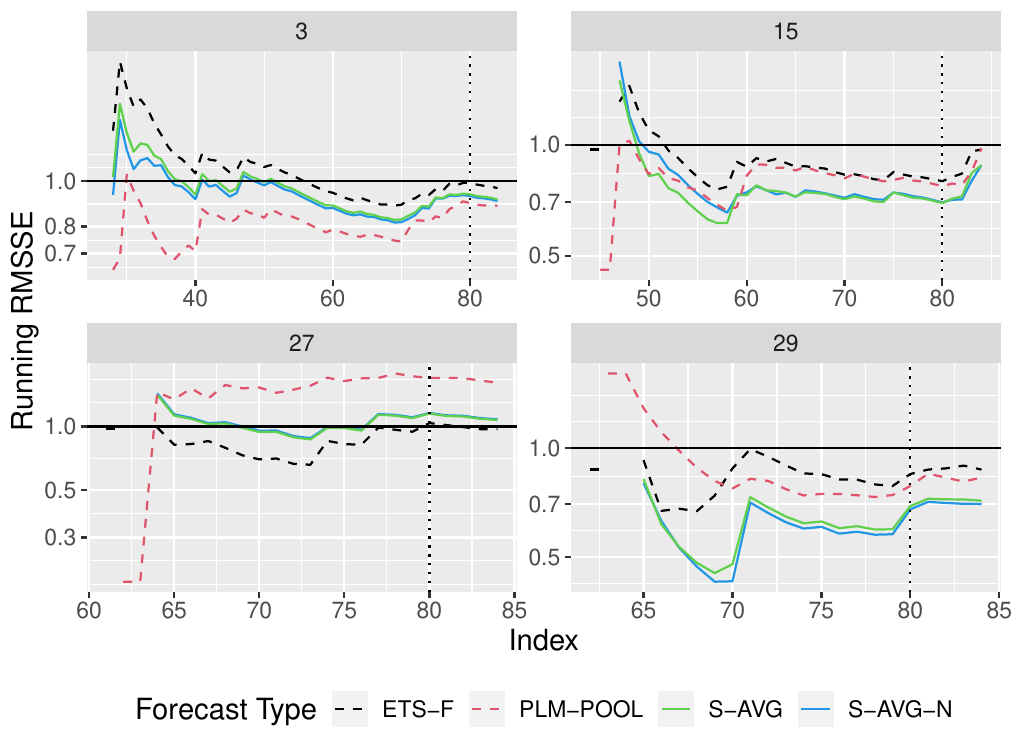}
    \caption{Running RMSSE with respect to random-walk forecasts for simple averaging.}
\end{figure}

\begin{figure}[!ht]
    \center
    \includegraphics[width=0.8\textwidth]{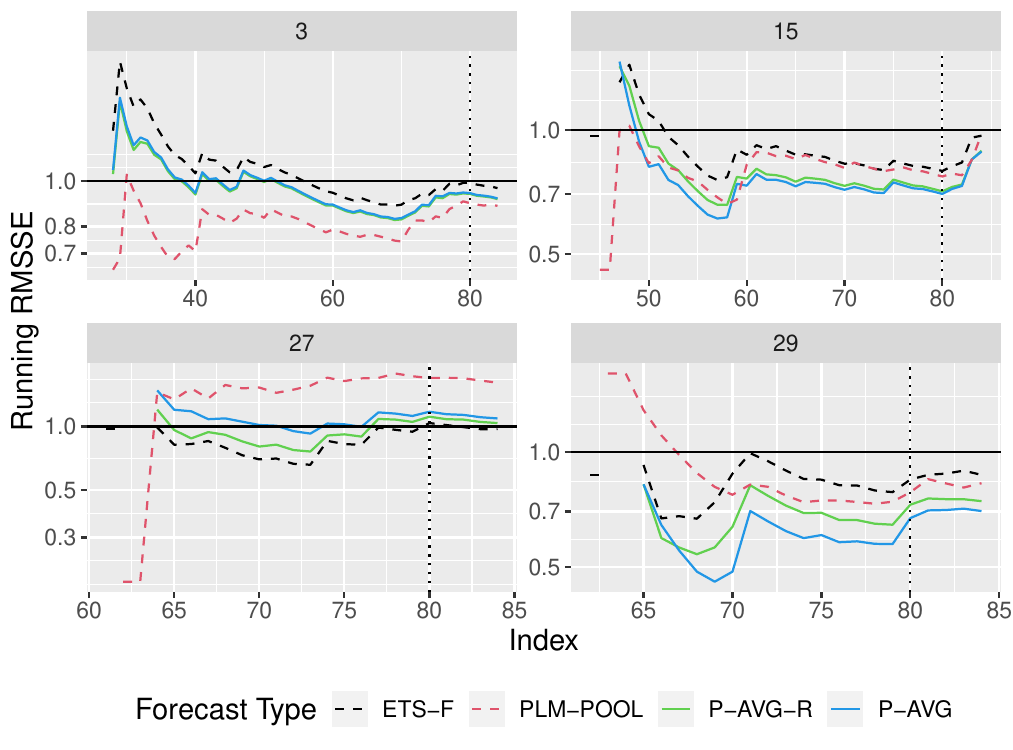}
    \caption{Running RMSSE with respect to random-walk forecasts for performance-based averaging.}
\end{figure}

\clearpage
%% -----------------------------------------------------------------------------

\bibliographystyle{plainnat}
\bibliography{refs.bib}

\end{document}

%% file: paper_tables/abbr.tex
\begin{table}[!ht]

\caption{\label{tab:abbr} Averaging notation.}
\centering
%\resizebox{\linewidth}{!}{
\begin{tabular}[t]{ll}
\toprule
Avg. Notation & Description \\
\midrule
G-AVG & Forecast the averaged time series\\
S-AVG & Model Avg. using simple weights\\
S-AVG-N & Model Avg. using simple weights and only neighbors\\
D-AVG & Model Avg. using distance-based weights\\
D-AVG-N & Model Avg. using distance-based weights and only neighbors\\ 
P-AVG & Model Avg. using performance-based weights\\
P-AVG-R & Model Avg. using refitted performance-based weights\\
S-NM-AVG & No Model Avg. using simple weights\\
D-NM-AVG & No Model Avg. using distance-based weights\\
\bottomrule
\end{tabular}
%}
\end{table}

%% file: paper_tables/datasets.tex
\begin{table}[!ht]

\caption{\label{tab:datasets} Summary of datasets.}
\centering
%\resizebox{\linewidth}{!}{
\begin{tabular}[t]{l|llll}
\toprule
Dataset Name & No. of Time Series & Forecast Horizon & Mean Length & Length Range\\
\midrule
Food Demand & 43 & 1-4 & 33 & 2 - 78 \\
M3 & 1,376 & 18 & 119 & 66 - 144 \\
Tourism & 366 & 24 & 299 & 91 - 333 \\
CIF 2016 & 72 & 6, 12 & 99 & 28 - 120 \\
Hospital & 767 & 12 & 84 & 84 - 84 \\
\bottomrule
\end{tabular}
%}
\end{table}

%% file: paper_tables/train_params.tex
\begin{table}[!ht]

\caption{\label{tab:train_params} Summary of training parameters. $T_0$ denotes the initial training timestep, $T_\text{train}$ the final training timestep and $L_\text{max}$ the maximum lag in the pooled model.}
\centering
%\resizebox{\linewidth}{!}{
\review{
\begin{tabular}[t]{l|rrr}
\toprule
Dataset Name & $T_0$ &  mean $T_\text{train}$ & $L_\text{max}$ \\
\midrule
Food Demand & 21 & 80 & 5 \\
M3 - Demographic &  92 & 308 & 22 \\
M3 - Finance &  94 & 411 & 22 \\
M3 - Industry &  63 & 172 & 22 \\
M3 - Macro &  73 & 378 & 22 \\
M3 - Micro &  87 & 111 & 22 \\
Tourism & 131 & 316 & 30 \\
CIF 2016 & 96 & 109 & 15  \\
Hospital & 58 & 72 & 15 \\
\bottomrule
\end{tabular}
}
\end{table}

%% file: paper_tables/multistep_eval_new.tex
\begin{table}

    \caption{\label{tab:multistep_mae}Multi-step MAE (p-value 0.127). The best results are in bold.}
    \centering
    \resizebox{\linewidth}{!}{
    \review{
    \begin{tabular}[t]{lrrrrrrrr}
    \toprule
    \multicolumn{1}{c}{ } & \multicolumn{2}{c}{M3} & \multicolumn{2}{c}{CIF2016} & \multicolumn{2}{c}{HOSPITAL} & \multicolumn{2}{c}{TOURISM} \\
    \cmidrule(l{3pt}r{3pt}){2-3} \cmidrule(l{3pt}r{3pt}){4-5} \cmidrule(l{3pt}r{3pt}){6-7} \cmidrule(l{3pt}r{3pt}){8-9}
    Model & Mean & Median & Mean & Median & Mean & Median & Mean & Median\\
    \midrule
    TSAVG & 713.580 & 510.021 & \textbf{545365.6} & 98.162 & 21.440 & \textbf{6.763} & 7109.038 & 985.100\\
    ETS & 736.036 & 507.657 & 642448.9 & \textbf{90.546} & 22.553 & 6.782 & 7391.902 & 959.128\\
    ARIMA & 704.468 & 487.471 & 568230.7 & 105.003 & 20.665 & 6.869 & 4184.401 & 811.594\\
    POOL & \textbf{678.815} & \textbf{485.306} & 1300457.8 & 219661.126 & \textbf{20.196} & 7.248 & \textbf{2258.485} & \textbf{501.032}\\
    \bottomrule
    \end{tabular}
    }}
    \end{table}
    \begin{table}
    
    \caption{\label{tab:multistep_rmse}Multi-step RMSE (p-value 0.127). The best results are in bold.}
    \centering
    \resizebox{\linewidth}{!}{
        \review{
    \begin{tabular}[t]{lrrrrrrrr}
    \toprule
    \multicolumn{1}{c}{ } & \multicolumn{2}{c}{M3} & \multicolumn{2}{c}{CIF2016} & \multicolumn{2}{c}{HOSPITAL} & \multicolumn{2}{c}{TOURISM} \\
    \cmidrule(l{3pt}r{3pt}){2-3} \cmidrule(l{3pt}r{3pt}){4-5} \cmidrule(l{3pt}r{3pt}){6-7} \cmidrule(l{3pt}r{3pt}){8-9}
    Model & Mean & Median & Mean & Median & Mean & Median & Mean & Median\\
    \midrule
    TSAVG & 860.793 & 604.199 & \textbf{632468.5} & 117.239 & 26.225 & 8.347 & 8632.070 & 1241.053\\
    ETS & 886.086 & 626.804 & 722426.3 & \textbf{107.687} & 27.336 & 8.324 & 8974.416 & 1222.899\\
    ARIMA & 845.835 & \textbf{588.559} & 641597.5 & 126.505 & 25.021 & \textbf{8.241} & 5377.097 & 1090.481\\
    POOL & \textbf{815.725} & 592.032 & 1456160.9 & 249604.091 & \textbf{24.309} & 8.709 & \textbf{2743.846} & \textbf{601.960}\\
    \bottomrule
    \end{tabular}}}
    \end{table}
    \begin{table}
    
    \caption{\label{tab:multistep_rmsse}Multi-step RMSSE (p-value 0.44). The best results are in bold.}
    \centering
    \resizebox{\linewidth}{!}{
        \review{
    \begin{tabular}[t]{lrrrrrrrr}
    \toprule
    \multicolumn{1}{c}{ } & \multicolumn{2}{c}{M3} & \multicolumn{2}{c}{CIF2016} & \multicolumn{2}{c}{HOSPITAL} & \multicolumn{2}{c}{TOURISM} \\
    \cmidrule(l{3pt}r{3pt}){2-3} \cmidrule(l{3pt}r{3pt}){4-5} \cmidrule(l{3pt}r{3pt}){6-7} \cmidrule(l{3pt}r{3pt}){8-9}
    Model & Mean & Median & Mean & Median & Mean & Median & Mean & Median\\
    \midrule
    TSAVG & \textbf{1.915} & 1.178 & 1.121 & 0.864 & 0.878 & \textbf{0.822} & 1.761 & 1.468\\
    ETS & 1.967 & 1.220 & \textbf{1.105} & \textbf{0.809} & 0.900 & 0.828 & 1.853 & 1.497\\
    ARIMA & 1.927 & \textbf{1.108} & 1.191 & 0.896 & \textbf{0.871} & 0.822 & 1.597 & 1.414\\
    POOL & 2.293 & 1.109 & 105106.814 & 2050.317 & 0.907 & 0.856 & \textbf{1.327} & \textbf{0.934}\\
    \bottomrule
    \end{tabular}}}
    \end{table}
    \begin{table}
    
    \caption{\label{tab:multistep_smape}Multi-step sMAPE (p-value 0.825). The best results are in bold.}
    \centering
    \resizebox{\linewidth}{!}{
        \review{
    \begin{tabular}[t]{lrrrrrrrr}
    \toprule
    \multicolumn{1}{c}{ } & \multicolumn{2}{c}{M3} & \multicolumn{2}{c}{CIF2016} & \multicolumn{2}{c}{HOSPITAL} & \multicolumn{2}{c}{TOURISM} \\
    \cmidrule(l{3pt}r{3pt}){2-3} \cmidrule(l{3pt}r{3pt}){4-5} \cmidrule(l{3pt}r{3pt}){6-7} \cmidrule(l{3pt}r{3pt}){8-9}
    Model & Mean & Median & Mean & Median & Mean & Median & Mean & Median\\
    \midrule
    TSAVG & 0.159 & 0.104 & 0.158 & \textbf{0.079} & 0.181 & 0.168 & 0.352 & 0.295\\
    ETS & 0.164 & 0.106 & \textbf{0.144} & 0.082 & 0.185 & 0.166 & 0.369 & 0.306\\
    ARIMA & 0.161 & \textbf{0.101} & 0.169 & 0.091 & \textbf{0.179} & \textbf{0.166} & 0.325 & 0.283\\
    POOL & \textbf{0.153} & 0.104 & 1.589 & 1.883 & 0.190 & 0.173 & \textbf{0.311} & \textbf{0.187}\\
    \bottomrule
    \end{tabular}}}
    \end{table}

%% file: paper_tables/onestep_eval_new.tex
\begin{table}

    \caption{\label{tab:onestep_mae}Cumulative one-step ahead MAE (p-value 0.564). The best results are in bold.}
    \centering
    \resizebox{\linewidth}{!}{
        \review{
    \begin{tabular}[t]{lrrrrrrrrrr}
    \toprule
    \multicolumn{1}{c}{ } & \multicolumn{2}{c}{M3} & \multicolumn{2}{c}{CIF2016} & \multicolumn{2}{c}{HOSPITAL} & \multicolumn{2}{c}{TOURISM} & \multicolumn{2}{c}{FOOD} \\
    \cmidrule(l{3pt}r{3pt}){2-3} \cmidrule(l{3pt}r{3pt}){4-5} \cmidrule(l{3pt}r{3pt}){6-7} \cmidrule(l{3pt}r{3pt}){8-9} \cmidrule(l{3pt}r{3pt}){10-11}
    Model & Mean & Median & Mean & Median & Mean & Median & Mean & Median & Mean & Median\\
    \midrule
    TSAVG & 484.660 & 322.573 & 301584.2 & 82.064 & 19.610 & \textbf{6.437} & 3647.229 & 770.753 & \textbf{8.523} & \textbf{6.917}\\
    ETS & 478.813 & 318.437 & 247264.0 & \textbf{76.378} & 19.635 & 6.491 & 3744.562 & 767.282 & 8.662 & 7.824\\
    ARIMA & \textbf{470.654} & \textbf{311.126} & \textbf{233306.1} & 79.007 & 19.040 & 6.495 & 3190.488 & 700.724 & 9.882 & 9.720\\
    POOL & 475.890 & 313.683 & 411428.1 & 32944.310 & \textbf{17.643} & 6.854 & \textbf{1543.199} & \textbf{461.658} & 9.438 & 7.387\\
    \bottomrule
    \end{tabular}}}
    \end{table}
    \begin{table}
    
    \caption{\label{tab:onestep_rmse}Cumulative one-step ahead RMSE (p-value 0.564). The best results are in bold.}
    \centering
    \resizebox{\linewidth}{!}{
        \review{
    \begin{tabular}[t]{lrrrrrrrrrr}
    \toprule
    \multicolumn{1}{c}{ } & \multicolumn{2}{c}{M3} & \multicolumn{2}{c}{CIF2016} & \multicolumn{2}{c}{HOSPITAL} & \multicolumn{2}{c}{TOURISM} & \multicolumn{2}{c}{FOOD} \\
    \cmidrule(l{3pt}r{3pt}){2-3} \cmidrule(l{3pt}r{3pt}){4-5} \cmidrule(l{3pt}r{3pt}){6-7} \cmidrule(l{3pt}r{3pt}){8-9} \cmidrule(l{3pt}r{3pt}){10-11}
    Model & Mean & Median & Mean & Median & Mean & Median & Mean & Median & Mean & Median\\
    \midrule
    TSAVG & 615.987 & 406.703 & 360462.5 & \textbf{93.481} & 23.985 & 8.080 & 4759.500 & 1005.795 & \textbf{10.112} & \textbf{8.417}\\
    ETS & 611.354 & 410.237 & 305246.7 & 93.493 & 23.991 & \textbf{7.992} & 4870.787 & 1008.100 & 10.193 & 8.720\\
    ARIMA & 598.826 & 395.709 & \textbf{301763.1} & 96.834 & 23.320 & 8.022 & 4119.036 & 929.144 & 11.552 & 11.365\\
    POOL & \textbf{597.438} & \textbf{389.837} & 451571.8 & 32944.598 & \textbf{21.610} & 8.357 & \textbf{2047.617} & \textbf{597.842} & 11.526 & 9.685\\
    \bottomrule
    \end{tabular}}}
    \end{table}
    \begin{table}
    
    \caption{\label{tab:onestep_rmsse}Cumulative one-step ahead RMSSE (p-value 0.241). The best results are in bold.}
    \centering
    \resizebox{\linewidth}{!}{
        \review{
    \begin{tabular}[t]{lrrrrrrrrrr}
    \toprule
    \multicolumn{1}{c}{ } & \multicolumn{2}{c}{M3} & \multicolumn{2}{c}{CIF2016} & \multicolumn{2}{c}{HOSPITAL} & \multicolumn{2}{c}{TOURISM} & \multicolumn{2}{c}{FOOD} \\
    \cmidrule(l{3pt}r{3pt}){2-3} \cmidrule(l{3pt}r{3pt}){4-5} \cmidrule(l{3pt}r{3pt}){6-7} \cmidrule(l{3pt}r{3pt}){8-9} \cmidrule(l{3pt}r{3pt}){10-11}
    Model & Mean & Median & Mean & Median & Mean & Median & Mean & Median & Mean & Median\\
    \midrule
    TSAVG & 0.904 & 0.805 & 0.885 & 0.802 & 0.824 & 0.800 & 1.418 & 1.320 & \textbf{0.875} & \textbf{0.748}\\
    ETS & 0.892 & 0.796 & \textbf{0.759} & 0.751 & 0.823 & 0.802 & 1.423 & 1.348 & 0.887 & 0.762\\
    ARIMA & \textbf{0.882} & \textbf{0.780} & 0.770 & \textbf{0.738} & \textbf{0.819} & \textbf{0.796} & 1.282 & 1.205 & 0.973 & 0.804\\
    POOL & 1.110 & 0.839 & 25562.026 & 280.377 & 0.842 & 0.824 & \textbf{0.929} & \textbf{0.816} & 1.354 & 0.939\\
    \bottomrule
    \end{tabular}}}
    \end{table}
    \begin{table}
    
    \caption{\label{tab:onestep_smape}Cumulative one-step ahead SMAPE (p-value 0.948). The best results are in bold.}
    \centering
    \resizebox{\linewidth}{!}{
        \review{
    \begin{tabular}[t]{lrrrrrrrrrr}
    \toprule
    \multicolumn{1}{c}{ } & \multicolumn{2}{c}{M3} & \multicolumn{2}{c}{CIF2016} & \multicolumn{2}{c}{HOSPITAL} & \multicolumn{2}{c}{TOURISM} & \multicolumn{2}{c}{FOOD} \\
    \cmidrule(l{3pt}r{3pt}){2-3} \cmidrule(l{3pt}r{3pt}){4-5} \cmidrule(l{3pt}r{3pt}){6-7} \cmidrule(l{3pt}r{3pt}){8-9} \cmidrule(l{3pt}r{3pt}){10-11}
    Model & Mean & Median & Mean & Median & Mean & Median & Mean & Median & Mean & Median\\
    \midrule
    TSAVG & 0.115 & 0.069 & 0.107 & 0.084 & \textbf{0.168} & \textbf{0.158} & 0.291 & 0.262 & \textbf{0.244} & \textbf{0.211}\\
    ETS & 0.114 & 0.068 & \textbf{0.097} & 0.082 & 0.169 & 0.160 & 0.288 & 0.265 & 0.251 & 0.214\\
    ARIMA & 0.114 & \textbf{0.065} & 0.101 & \textbf{0.080} & 0.169 & 0.161 & 0.267 & 0.241 & 0.301 & 0.242\\
    POOL & \textbf{0.114} & 0.066 & 1.338 & 1.592 & 0.176 & 0.169 & \textbf{0.206} & \textbf{0.158} & 0.274 & 0.230\\
    \bottomrule
    \end{tabular}}}
    \end{table}

%% file: paper_tables/best_k.tex
\begin{table}[!ht]

\caption{\label{tab:no_neighbors}Optimal neighborhood sizes based on TSCV.}
\centering
\begin{tabular}[t]{l|rrrr}
\toprule
\multicolumn{1}{c}{ } & \multicolumn{4}{c}{Individuals} \\
\cmidrule(l{3pt}r{3pt}){2-5}
Forecast Type & 3 & 15 & 27 & 29\\
\midrule
S-NM-AVG & 1 & 5 & 20 & 20\\
D-NM-AVG & 3 & 5 & 20 & 10\\
S-AVG & 3 & 1 & 20 & 5\\
S-AVG-N & 3 & 1 & 20 & 5\\
D-AVG & 3 & 1 & 20 & 3\\
D-AVG-N & 3 & 1 & 20 & 5\\
G-AVG & 1 & 10 & 1 & 10\\
P-AVG-R & 3 & 1 & 10 & 5\\
P-AVG & 3 & 1 & 20 & 5\\
\bottomrule
\end{tabular}
\end{table}

%% file: paper_tables/rmsse_fridges.tex
%\begin{table}[!ht]

%\caption{\label{tab:tbl:rmsse_fridges}Test RMSSE for benchmark models and best averaging models. The numbers in parenthesis are the standard errors. The best model is in bold.}
%\centering
%\begin{tabular}[t]{l|ll|ll}
%\toprule
%ID & ETS & PLM-POOL & Best Avg. & RMSSE\\
%\midrule
%3 & 0.967 (0.168) & \textbf{0.886} (0.139) & D-AVG-N & 0.894 (0.135)\\
%15 & 0.970 (0.258) & 0.979 (0.292) & D-AVG & \textbf{0.875} (0.212)\\
%27 & \textbf{0.972} (0.309) & 1.609 (0.525) & P-AVG-R & 1.037 (0.319)\\
%29 & 0.873 (0.163) & 0.830 (0.170) & D-AVG-N & \textbf{0.698} (0.154)\\
%\bottomrule
%\end{tabular}
%\end{table}

\begin{table}
    \caption{\label{tab:tbl:errors_fridges}Test errors for benchmark models and best averaging method (P-AVG-R). The best models are in bold.}
    \centering
    \resizebox{\linewidth}{!}{
    \begin{tabular}{lrrrrrrrrrrrrrrrr}
    \toprule
    \multicolumn{1}{c}{ } & \multicolumn{4}{c}{sMAPE} & \multicolumn{4}{c}{MAE} & \multicolumn{4}{c}{RMSE} & \multicolumn{4}{c}{RMSSE} \\
    \cmidrule(l{3pt}r{3pt}){2-5} \cmidrule(l{3pt}r{3pt}){6-9} \cmidrule(l{3pt}r{3pt}){10-13} \cmidrule(l{3pt}r{3pt}){14-17}
    Model & 3 & 15 & 27 & 29 & 3 & 15 & 27 & 29 & 3 & 15 & 27 & 29 & 3 & 15 & 27 & 29\\
    \midrule
    TSAVG & \textbf{0.147} & 0.312 & \textbf{0.083} & 0.296 & \textbf{4.094} & 12.184 & \textbf{13.177} & 6.699 & \textbf{5.353} & 12.946 & \textbf{16.033} & \textbf{6.962} & \textbf{0.439} & 1.901 & \textbf{0.609} & \textbf{0.668}\\
    ETS & 0.163 & 0.298 & 0.097 & \textbf{0.294} & 4.531 & \textbf{11.656} & 15.466 & \textbf{6.566} & 5.559 & \textbf{12.321} & 16.883 & 7.103 & 0.456 & \textbf{1.809} & 0.642 & 0.682\\
    ARIMA & 0.222 & 0.332 & 0.097 & 0.367 & 6.127 & 13.299 & 15.467 & 7.918 & 6.765 & 14.247 & 16.885 & 8.832 & 0.554 & 2.092 & 0.642 & 0.848\\
    POOL & 0.225 & \textbf{0.297} & 0.169 & 0.310 & 5.073 & 12.097 & 22.788 & 7.160 & 6.873 & 12.890 & 30.948 & 7.941 & 0.595 & 1.870 & 1.124 & 0.893\\
    \bottomrule
    \end{tabular}}
 \end{table}